\documentclass[11pt]{article}
\usepackage[utf8]{inputenc}
\usepackage[margin=1in]{geometry}
\usepackage{epsfig}
\usepackage{times}
\usepackage{amsmath, amsthm, amssymb}
\usepackage{amsmath,blkarray,booktabs,bigstrut}
\usepackage{verbatim} 
\usepackage{bbm}
\usepackage{mathtools}

\usepackage[linesnumbered, ruled,vlined]{algorithm2e}

\usepackage{algpseudocode}
\SetKwComment{Comment}{// }{}
\usepackage{hyperref}
\hypersetup{
    linkcolor=red,
    filecolor=magenta,      
    urlcolor=cyan,
    pdftitle={Your Document Title},
    pdfpagemode=FullScreen,
}
\usepackage{natbib}
\usepackage{cleveref}
\usepackage{nicefrac}
\Crefname{algorithm}{Algorithm}{Algorithms}
\newcommand{\eat}[1]{}

\usepackage{color}

\newcommand{\Problem}[1]{$(k,\ell)$-\textsc{Trading Prophet Problem}}
\newcommand{\ProblemTwo}[1]{$(k,\ell, \ell')$-\textsc{Trading Prophet Problem}}

\newcommand{\MatroidSetting}[1]{$\mathcal{M}$-\textsc{Trading Prophet Problem}}

\makeatletter
\newtheorem*{rep@theorem}{\rep@title}
\newcommand{\newreptheorem}[2]{%
\newenvironment{rep#1}[1]{%
 \def\rep@title{#2 \ref{##1}}%
 \begin{rep@theorem}}%
 {\end{rep@theorem}}}
\makeatother

\newtheorem{theorem}{Theorem}[section]
\newtheorem{lemma}{Lemma}[section]
\newtheorem{corollary}{Corollary}[section]
\newtheorem{definition}{Definition}[section]
\newtheorem{fact}{Fact}[section]
\newtheorem{Observation}{Observation}[section]
\newtheorem{proposition}{Proposition}[section]

\author{
Surbhi Rajput\thanks{Indian Institute of Technology, Delhi, India} \and Ashish Chiplunkar\thanks{Indian Institute of Technology, Delhi, India, \url{https://www.cse.iitd.ac.in/\~ashishc/}} \and Rohit Vaish\thanks{Indian Institute of Technology, Delhi, India, \url{https://www.cse.iitd.ac.in/\~rvaish/}}
}
\title{Trading Prophets: How to Trade Multiple Stocks Optimally}
\date{\vspace{-5ex}}
\begin{document}

\maketitle

\begin{abstract}
\small\baselineskip=9pt In the (single stock) \emph{trading prophet} problem formulated by~\citet{CCD+23trading}, an online algorithm observes a sequence of prices of a stock. At each step, the algorithm can either buy the stock by paying the current price if it doesn't already hold the stock, or it can sell the currently held stock and collect the current price as a reward. The goal of the algorithm is to maximize its overall profit. \citet{CCD+23trading} showed that the optimal competitive ratio for this problem is $\nicefrac{1}{2}$ when the stock prices are identically and independently distributed.

In this work, we generalize the model and the results of~\citet{CCD+23trading} by allowing the algorithm to trade multiple stocks. First, we formulate the \ProblemTwo{}, wherein there are $k$ stocks in the market, and the online algorithm can hold up to $\ell$ stocks at any time, where $\ell \leq k$. The online algorithm competes against an offline algorithm that can hold at most $\ell' \leq \ell$ stocks at any time. Under the assumption that prices of different stocks are independent, we show that, for any $\ell$, $\ell'$, and $k$, the optimal competitive ratio of \ProblemTwo{} is $\min\left\{\frac{1}{2},\frac{\ell}{k}\right\}$.

We further introduce the more general \MatroidSetting{} over a matroid $\mathcal{M}$ on the set of $k$ stocks, wherein the stock prices at any given time are possibly correlated (but are independent across time). The algorithm is allowed to hold only a feasible subset of stocks at any time. We prove a tight bound of $\frac{1}{1+d}$ on the competitive ratio of the \MatroidSetting{}, where $d$ is the \textit{density} of the matroid (refer \Cref{def_rank_density}).

We then consider the non-i.i.d.\ random order setting over a matroid, wherein stock prices drawn independently from $n$ potentially different distributions are presented in a uniformly random order. In this setting, we achieve a competitive ratio of at least $\frac{1}{1+d} - \mathcal{O} \left(\frac{1}{n} \right)$, where $d$ is the density of the matroid, matching the hardness result for i.i.d.\ instances as $n$ approaches $\infty$.

Our analysis of the above problems is based on the following key insights. First, any algorithm can be simulated by one that, on each time step, sells \emph{all} its currently held stocks before buying a suitable subset of stocks. Second, we prove that the general problem reduces to a restriction where the expected price of every stock is zero.
Third, we reduce the problem in the random order non-i.i.d.\ setting to the i.i.d. setting by leveraging the fact that the outcome of sampling two objects without replacement from a large set is almost identically distributed as the outcome of sampling with replacement. 
\end{abstract}

\section{Introduction}
Consider a trader named Alex, who wants to invest her capital in the stock market and maximize her returns. In the stock market, prices of stocks fluctuate due to various factors including market demand, company performance, and economic indicators. Alex monitors these price shifts in real time, with the price of each stock being available as a sequence over time in an online manner. Alex would have loved to know the future behavior of stock prices, so that she could trade optimally and maximize her profit. However, the online setting, which mirrors real-world trading conditions, presents a challenge. Alex can't see the future prices, so she is not necessarily able to trade optimally. The framework of online computation and competitive analysis~\citep{SleatorT85} captures Alex's dilemma: Given the current prices and without knowing their future behavior, how should Alex decide which stocks to trade at every time step? Alex must run an \textit{online algorithm} -- one whose current output depends only on the current and the past inputs. Informally, such an algorithm is said to be \emph{$\alpha$-competitive} if it guarantees an expected payoff at least $\alpha$ times the optimal payoff.

If Alex is completely uninformed about future outcomes, it is impossible for \emph{any} algorithm to give her a non-trivial competitive guarantee. Indeed, if the input is adversarial, every stock Alex buys crashes, and every stock she doesn't buy soars. Thankfully for Alex, such adversarial behavior of prices does not arise in practice. In fact, she knows in advance the probability distributions of future stock prices. This setting is referred to in the literature as the \emph{prophet inequality} setting~\citep{KrengelS, HillK}.\footnote{In the classic single-choice prophet inequality problem~\citep{HKS07automated,L17economic,EJL+17prophet,CFJ+19recent,PT22order}, an online algorithm observes a sequence of random variables with known distributions. The value of a random variable is realized upon its arrival, and the algorithm should either accept it (in which case the process stops) or reject it (in which case the next variable is observed). The goal is to maximize the selected value while competing against an offline adversary, also referred to as the \textit{prophet}, that observes all realizations at once and chooses the maximum out of them. This problem has also been studied under feasibility constraints beyond single choice; see, for example, \citep{A14bayesian, KleinbergW19, EHK+18prophet}.}

The recent work of~\citet{CCD+23trading} considered the simplest case of trading one stock and introduced the \emph{Trading Prophet} problem. In this problem, an online algorithm is given a sequence of $n$ prices, each sampled independently from the same probability distribution that is known to the algorithm. 
The algorithm trades a single (indivisible) stock and is limited to holding at most one unit of stock at any given time. At each time step, the algorithm must decide whether to \emph{purchase} the stock at the current price (provided it does not hold the stock already) or to \emph{sell} the stock at the current price (provided it already holds the stock). Correa et al.\ proved that
the strategy of buying when the price falls below the \emph{median} and selling otherwise is $\nicefrac{1}{2}$-competitive\footnote{The median of the distribution of a random variable $X$ is a number $\mu$ such that $\Pr[X<\mu]\leq \nicefrac{1}{2}$ and $\Pr[X>\mu]\leq \nicefrac{1}{2}$.}. Moreover, they also proved that no online strategy can achieve a competitive ratio strictly better than $\nicefrac{1}{2}$. Their analysis extends to the scenario where the algorithm can trade multiple units of a single stock.

\subsection*{Our Contributions.}

While the single stock problem is instructive, it may not be representative of real-world trading scenarios. It is more natural to imagine that an investor trades more than one stock at any given time. Motivated by this, in this paper, we generalize the Trading Prophet model to accommodate the idea of holding \emph{multiple} stocks. We believe this model better reflects the real-world trading dynamics.

\paragraph{Generalizing the Trading Prophet problem.}
We introduce the \Problem{}
in which we are given $k$ different stocks and the algorithm can hold up to $\ell$ stocks (and at most one unit of any given stock; this is without loss of generality) at any time. The price of every stock at any time step is drawn independently from a known distribution. Note that the distributions of prices of the $k$ stocks can be different. For any given time step, the prices of stocks need not be independently distributed. However, we continue with the assumption of \cite{CCD+23trading} that across time steps these prices are independent.
At every time step, an algorithm must decide which subset of its currently owned stocks to sell, and which subset of stocks not currently owned to buy, subject to the constraint that the algorithm holds at most $\ell$ stocks at any given time. 
We establish tight bounds on the competitive ratio of the \Problem{}. Specifically, we prove,
\begin{theorem}\label{k-l:theorem}
For every $k$ and $\ell$, there exists an online algorithm for the \Problem{} that achieves a competitive ratio at least $\min\left\{\frac{1}{2},\frac{\ell}{k}\right\}$ when the stock prices at any given time are independent. Moreover, no algorithm can guarantee a strictly better competitive ratio.
\end{theorem}

\paragraph{Extension to resource augmentation setting.}
Our result in Theorem~\ref{k-l:theorem} extends to the \textit{resource augmentation} setting~\citep{Roughgarden}. 
In this setting, the online algorithm competes with the offline algorithm having fewer resources than the online algorithm. Specifically, the offline algorithm is restricted to hold up to $\ell'$ stocks at any time, where $1\leq\ell'\leq\ell$.
We call this generalization the \ProblemTwo{}. Even in this setting, we prove that the tight bound from Theorem~\ref{k-l:theorem} continues to hold, making it a special case of the following two theorems for $\ell'=\ell$.
    \begin{theorem}
    \label{hardness_result}
        For every $k$, ${\ell}$ and $\ell'$, no algorithm for \ProblemTwo{} can achieve a competitive ratio greater than $\min\left\{\frac{1}{2},\frac{\ell}{k}\right\}$ when the stock prices at any given time are independent.
    \end{theorem}
    
\begin{theorem}
\label{theorem1}
For every $k$, $\ell$, and $\ell'$, there exists an algorithm that is $\min\left\{\frac{1}{2},\frac{\ell}{k}\right\}$-competitive algorithm for the \ProblemTwo{} when the stock prices at any given time are independent.
\end{theorem}

One might intuitively expect that the competitive ratio will improve if the adversary's power is reduced by lowering $\ell'$. Somewhat surprisingly, the optimal competitive ratio established by Theorems \ref{hardness_result} and \ref{theorem1} \emph{does not} depend on $\ell'$. Our instance to force the hardness result is designed so that, with high probability, at most one stock is profitable at any given time. The offline algorithm can easily spot such a stock and buy it (even when $\ell'=1$), while the online algorithm is unlikely to guess it correctly, even though the possibility of holding up to $\ell$ stocks essentially gives the algorithm $\ell$ opportunities to guess the right stock.
Another interesting takeaway is that for any fixed $k$, as long as the online algorithm can hold at least 50\% of all stocks (i.e., $\ell \geq k/2$), increasing the number of stocks held \emph{does not} affect the competitive ratio. We believe that these insights are not obvious a priori.

\paragraph{Generalization to correlated prices.}
For the \Problem{}, we assumed that the prices are independent across time and independent across stocks in Theorem~\ref{k-l:theorem}. We generalize this model to deal with \emph{correlated} stock prices to better relate with real-world trading scenarios. We solve for the setting where the stock prices are possibly correlated across stocks for any given time step, but independent across time. For this correlated setting of the \Problem{}, we prove the tight bound of $\frac{\ell}{\ell+k}$ (Corollary~\ref{coro_hardness} and Corollary~\ref{coro_alg}) as an easy consequence of Theorems \ref{hardness_res_matroid} and \ref{thm:matroid:ratio} stated below.

\paragraph{Generalization to matroid settings.}
The constraint of holding up to $\ell$ stocks at any given time is not sufficient to capture traders' requirement to keep their portfolios sufficiently diversified. 
For instance, suppose the stocks are categorized into different types. To maintain diversity, a trader wants to hold at most $\ell_i$ stocks of the $i$'th type in addition to holding at most $\ell$ stocks, at all times.
Such requirements are readily captured by combinatorial structures called \emph{matroids} (refer Definition~\ref{def_matroids}). Motivated by this, we consider the substantial generalization of the trading prophet problem to matroid constraints.
We introduce the \MatroidSetting{}, over the matroid $\mathcal{M}$ on the set of stocks. Here, the trader is allowed to hold a subset of stocks provided that the subset is a \emph{feasible\footnote{The matroid literature uses the term `independent' instead of `feasible'. We prefer to use the latter to avoid potential ambiguity resulting from the usage of `independent' in the context of matroids as well as random variables.} set} of $\mathcal{M}$.
We establish that the competitive ratio of the trading prophet problem on a matroid is determined by a property of the matroid called its \emph{density} (refer \Cref{def_rank_density}). Specifically, we prove,

\begin{theorem}
\label{hardness_res_matroid}
    Let $\mathcal{M}$ be an arbitrary matroid and $d$ be the density of $\mathcal{M}$. Then there does not exist an algorithm for \MatroidSetting{} whose competitive ratio is greater than $\frac{1}{1+d}$.
\end{theorem}
\begin{theorem}
\label{thm:matroid:ratio}
    Let $\mathcal{M}$ be an arbitrary matroid and $d$ be the density of $\mathcal{M}$. Then there exists an algorithm for the \MatroidSetting{} whose competitive ratio is at least $\frac{1}{1+d}$.
\end{theorem}

\paragraph{Generalization to non-i.i.d.\ Random order setting.}
In all the aforementioned settings, we consider the assumption that stock prices are independently and identically distributed (i.i.d.) over time. However, we also consider the non-i.i.d.\ setting, where prices are drawn from $n$ potentially different distributions and presented in a random order. In this setting, we prove, 

\begin{theorem}
\label{thm:random_order}
    For the non-i.i.d.\ random-order variant of the trading prophet problem over a matroid, there exists an algorithm whose competitive ratio is at least $\frac{1}{1+d} - \frac{2}{n}$,
    where $d$ denotes the density of the matroid.
\end{theorem}
This result implies that the competitive ratio in the non-i.i.d.\ random-order setting approaches the competitive ratio of the i.i.d.\ case as the time horizon $n$ grows unbounded.

\paragraph{Insights and techniques.}
We make the following crucial observations which help us get a much simpler analysis of \citet{CCD+23trading}'s result and provide a convenient way to analyze all our generalizations. 
\begin{itemize}
    \item The decision of opting to not sell a previously held stock at a given time step equates to selling the stock and repurchasing it at the same price and at the same time step, leaving the net profit unchanged. This simple yet crucial observation enables us to express the net profit of the algorithm over $n$ time steps as the sum of $(n-1)$ i.i.d.\ random variables. This reduces our analysis to relating the expected \emph{per-time-step} profit of the online and offline algorithms.
    \item In order to establish competitive guarantees, it suffices to consider instances where the expected value of each stock's price is zero, further simplifying our analysis.
    \item Consider two random experiments conducted on a set of size $n$. In the first experiment, two objects are sampled uniformly at random without replacement, while in the second, two objects are sampled uniformly at random with replacement. As $n$ approaches $\infty$, the distributions of the outcomes in both experiments become increasingly similar. 
    This key idea helps us to achieve a reduction from the random order non-i.i.d.\ setting to the i.i.d.\ setting by pretending as if the prices are drawn independently from the mixture distribution of the given distributions.
\end{itemize}

\paragraph{Organization of the paper}
In~\cref{sec:prelims}, we introduce the necessary preliminaries, notation, and relevant details for the problem and its analysis. \Cref{sec:hardness_result,sec:algorithmic} focus on the setting where stock prices are independently and identically distributed (i.i.d.) over time. Specifically, \cref{sec:hardness_result} establishes hardness guarantees, while \cref{sec:algorithmic} provides algorithmic guarantees by deriving the competitive ratio.

In~\cref{sec:random_order}, we discuss the setting where the prices are drawn independently but not necessarily from identical distributions and are presented in a uniformly random order. Here, we show a reduction from the non-i.i.d.\ random order setting to the i.i.d.\ setting.

\section{Preliminaries}
\label{sec:prelims}
For any positive integer $n$, let $[n]$ denote the set $\{1,2,\dots,n\}$.

\subsection{Matroids}

\begin{definition}
\label{def_matroids}
    A \textbf{matroid}\footnote{For a comprehensive discussion on matroids, we refer the reader to~\citep{oxley}.} \( \mathcal{M} \) is a pair \( (E, \mathcal{I}) \), where \( E \) is a finite set, called the \textbf{ground set}, and \( \mathcal{I} \) is a family of subsets of \( E \), called the set of \textbf{feasible sets}, with the following properties:
\begin{itemize}
    \item The empty set is feasible, i.e., \( \emptyset \in \mathcal{I} \).
    \item Every subset of a feasible set is feasible, i.e., $\mathcal{I}$ is downward-closed.
    \item If \( I \) and \( J \) are feasible sets and \( |I| > |J| \), there exists an element \( e \in I \setminus J \) such that \( J \cup \{e\} \) is feasible. 
\end{itemize}
We say a matroid $\mathcal{M}$ is
\begin{itemize}
    \item \textbf{loopless} if every singleton subset of $E$ is feasible, i.e., for every \( e \in E \), \( \{e\} \in \mathcal{I} \) (an element $e$ is called a \textbf{loop} if $\{e\} \notin \mathcal{I}$).
    \item the $r$-\textbf{uniform matroid} over $E$ if for every set $A \subseteq E$, $A$ is feasible if and only if $|A| \leq r$. 
\end{itemize}
\end{definition}

\begin{definition}
\label{def_rank_density}
    The \textbf{rank} of a matroid is the size of the largest feasible set in the matroid.
The \textbf{rank function} of a matroid \( \mathcal{M} \) is a function \( \texttt{rk} : 2^E \to \mathbb{N} \) that maps each subset \( S \subseteq E \) to the maximum size of a feasible subset of \( S \). 
An element $e$ is said to be \textbf{spanned} by a set \( S \subseteq E \) in a matroid \( \mathcal{M} \) if the rank of \( S \) is equal to the rank of \( S \cup \{e\} \). 

The \textbf{density} \( d \) of a loopless matroid \( \mathcal{M} \) is defined as the maximum value of the ratio \( \nicefrac{|X|}{\texttt{rk}(X)} \) over all non-empty subsets \( X \) of the ground set \( E \) \citep{soto2013matroid}. Formally,
\[d = \max_{\emptyset \neq X \subseteq E} \frac{|X|}{\texttt{rk}(X)}\text{.}\]

\end{definition}

\begin{Observation}
\label{observation1}
    The density of the $r$-uniform matroid on a set of $n$ elements (where $r\leq n$) is $\nicefrac{n}{r}$.
\end{Observation}

\paragraph{}  
The \textbf{maximum weight feasible set problem} over a matroid \( \mathcal{M} = (E, \mathcal{I}) \) is defined as follows: Given a weight function \( w : E \to \mathbb{R}_+ \), find a feasible set \( I \in \mathcal{I} \) that maximizes the total weight. 
A simple greedy idea to solve the maximum weight feasible set problem, known as Kruskal's algorithm, is given by \Cref{alg:Kruskal}. The following is a folklore result.

\begin{fact}
    The \emph{Kruskal's algorithm} stated in \Cref{alg:Kruskal} returns the maximum weight feasible set over a matroid constraint. 
\end{fact}

\begin{algorithm}[t]
\DontPrintSemicolon
    \KwIn{A matroid $\mathcal{M} = (E, \mathcal{I})$, and a weight function $w : E \to \mathbb{R}$.}
    \KwOut{A set $I \in \mathcal{I}$ that maximizes $w(I) = \sum_{e \in I} w(e)$.}
    Sort the elements of $E$ in non-increasing order of weights and call them $e_1, \ldots, e_n$.\;
    \( H \gets \emptyset \).\;
    \For{$i = 1$ \KwTo $n$}{
        \If{$w(e_i) < 0$} {
            \Comment{This implies $w(e_{i+1}), \ldots, w(e_{n})$ are all negative.}
            Break\;
        }
        \If{$H \cup \{e_i\} \in \mathcal{I}$}{
            \Comment{Equivalently, if $e_i$ is not spanned by $\{e_1, \ldots, e_{i-1}\}$.}
            $H \gets H \cup \{e_i\}$\;
        }
    }
    \textbf{Return} $H$\;
    \caption{Kruskal's Algorithm for Maximum Weight Feasible Set}
    \label{alg:Kruskal}
\end{algorithm}

\subsection{Trading Prophets}
\paragraph{\MatroidSetting{}.} 
An instance of the \MatroidSetting{} is specified by a matroid $\mathcal{M}$ on the ground set $\{1, \ldots, k\}$ for some $k \in \mathbb{N}$, the joint distribution of $k$ real-valued random variables specified by its cumulative distribution function\footnote{The \textit{cumulative distribution function} of a vector $(X_1,\ldots,X_k)$ of random variables is the function $F$ given by $F(x_1, \ldots, x_k) = \Pr[X_1\leq x_1\wedge \cdots \wedge X_k\leq x_k]$.} (CDF) $F:\mathbb{R}^k\longrightarrow[0,1]$, and a positive integer $n$, the length of the time horizon. The random variable $X^t$, taking values from $\mathbb{R}^k$, denotes the vector of stock prices at time $t\in[n]$. Stock prices for different timestamps are independent, but the components $\{X^t_s\mid s\in[k]\}$ of each $X^t$ are not necessarily independent. If the components of each $X^t$ are also independent, we call such an instance \emph{independently distributed}. Formally,
\begin{definition}
    \label{independently_distributed}
    An instance of \MatroidSetting{} is called an \textbf{independently distributed} instance if there exist CDFs $F_1, \ldots, F_k : \mathbb{R}\longrightarrow[0,1]$ such that the CDF $F$ of the instance satisfies $F(x_1, \ldots, x_k) = \prod_{s=1}^k F_s (x_s)$. 
\end{definition}
At each step $t \in [n]$, the vector $X^t$ of prices of all stocks is announced. The random variable, $X^t_s$, denotes the stock price $s$ at time step $t$. 
In response, the algorithm should decide the following at each step:
\begin{itemize}
    \item which subset of previously held stocks should be \emph{sold}, and 
    \item which subset of stocks not currently held should be \emph{bought},
\end{itemize}
subject to the constraint that the set of stocks held by the algorithm is a feasible set of $\mathcal{M}$ at all times. The algorithm starts with an empty set of stocks and at the end of the time horizon (i.e., after $n$ time steps), the algorithm sells all the held stocks. Thus, the \emph{profit} earned by the algorithm at any time $t$ equals the amount received due to the sale of stocks minus the amount paid to buy stocks, and the \emph{overall profit} is the sum of profits across all time steps $t \in [n]$. Our goal is to maximize the expected overall profit. Since an algorithm can never buy a stock which is a loop of the matroid, we assume that $\mathcal{M}$ is a loopless matroid.

To simplify our presentation, we work with an equivalent setup where, in every time step, the algorithm sells all currently held stocks and then buys a subset of stocks $S$ such that $S$ is a feasible set in $\mathcal{M}$ -- not selling a stock at time $t$ generates the same profit as selling the stock at time $t$ and then buying the same stock at the same price at time $t$. 

\paragraph{\Problem{}.} This is the restriction of \MatroidSetting{} where $\mathcal{M}$ is the $\ell$-uniform matroid over a ground set of size $k$.

\paragraph{\ProblemTwo{}.} This is the \Problem{} extended to the resource augmentation setting. In this problem, just as in the \Problem{}, we are given the CDF of a joint distribution of $k$ stocks, integers $\ell, \ell' \in \mathbb{N}$, and a positive integer $n$. An online algorithm for the \ProblemTwo{} is allowed to hold at most $\ell$ stocks and its goal, as before, is to maximize its profit across $n$ time steps. The online algorithm competes against an offline algorithm that is allowed to hold at most $\ell' \leq \ell$ stocks. Thus, when $\ell=\ell'$, the \ProblemTwo{} becomes the \Problem{}.

\paragraph{}
An \emph{online algorithm} must decide which subset of stocks to buy and sell at each time step based only on the history and the current prices of stocks. On the other hand, an \emph{offline algorithm} has access to all the realizations of stock prices. In general, an online algorithm will fail to obtain as much profit as the offline optimal algorithm due to its inability to see the future trend of stock prices.
We use the framework of competitive analysis~\citep{SleatorT85} to quantify the performance of an online algorithm relative to that of the optimal offline algorithm. The \textit{competitive ratio} of an online algorithm is defined as follows.

\begin{definition}
The \textbf{competitive ratio} of an online (maximization) algorithm is the infimum over all instances of the ratio of the expected profit of the online algorithm to that of an optimal offline solution.
An online algorithm is said to be $\alpha$-competitive if its competitive ratio is at least $\alpha$.
\end{definition}

\paragraph{Notation.}
We frequently refer to the following notation in our analysis:
\begin{itemize}
\item $\texttt{top}_\mathcal{M}(w)$ denotes the weight of the maximum weight feasible set of the matroid $\mathcal{M}=(E,\mathcal{I})$ with respect to the weight function $w:E\longrightarrow\mathbb{R}_+$, and $\texttt{top}_\ell(w)$ denotes $\texttt{top}_\mathcal{M}(w)$ where $\mathcal{M}$ is the $\ell$-uniform matroid. Therefore, $\texttt{top}_\ell(w)$ is the sum of the $\ell$ largest numbers out of $w_1,\ldots,w_k$ if at least $\ell$ of them are positive, and the sum of all positive numbers in $w_1,\ldots,w_k$ if fewer than $\ell$ of them are positive.
\item We assume that the probability distribution with CDF $F$ over $\mathbb{R}^k$ has a well-defined expectation.
$\mu = (\mu_1, \ldots, \mu_k)$ denotes this expectation.
\item The ordered pair $(X', X)$ denotes two samples drawn independently from the distribution having CDF $F$. $X'$ and $X$ are only used in the analysis.
\item The expression $x^+$ denotes $\max(0,x)$.
\end{itemize}

\begin{lemma}[Optimal offline profit]\label{lem_perday_opt}
    For every instance of the \MatroidSetting{}, the expected overall profit of the optimal offline algorithm is $$(n-1)\mathbb{E}\left[\texttt{\textup{top}}_{\mathcal{M}} \left(X' - X \right)\right].$$
\end{lemma}

\begin{proof}
 Recall the equivalent setup where at time step $t$, the algorithm sells all currently held stocks and then buys a set of stocks. 
 Suppose we buy a subset $S$ of stocks at time $t$ and sell it at time $t+1$, then our net profit will be $\sum_{s \in S} (X^{t+1}_s-X^t_s) $. Therefore, at every time $t$, the optimal algorithm buys the maximum weight feasible subset of stocks with respect to the weight function $X^{t+1}-X^t$. 
 Thus, the resulting profit of the algorithm at time $t$ is $\texttt{top}_{\mathcal{M}} \left( X^{t+1}-X^t\right)$. 
Since the random variables $X^{t+1}$ and $X^t$ are two independent samples drawn from the distribution having CDF $F$, the distribution of $(X^{t+1},X^t)$ is same as $(X',X)$. Therefore, the expected profit is $\mathbb{E}\left[\texttt{top}_{\mathcal{M}} \left( X'-X \right)\right]$. Summing over time steps $t\in[n-1]$, we get the result.
\end{proof}

\begin{lemma}[Optimal online profit]\label{lem_perday_alg}
    For every instance of the \MatroidSetting{}, the expected overall profit of the optimal online algorithm is $$(n-1) \, \mathbb{E}\left[\texttt{\textup{top}}_{\mathcal{M}} \left( \mu - X \right) \right].$$
    Furthermore, this profit is attained by \Cref{alg:Online}.
\end{lemma}

\begin{proof}
Consider the execution of an arbitrary online algorithm at timestamp $t$, after selling all the previously held stocks. If the algorithm buys stock $s$ (for the price $X^t_s$), the expected price at which it will get sold at time $t+1$ is $\mathbb{E}[X^{t+1}_s]=\mu_s$, because $X^{t}_s$ and $X^{t+1}_s$ are independent. Thus, conditional on the current price $X^t_s$, the expected profit gained from stock $s$ is $\mu_s-X^t_s$ if stock $s$ is bought (and zero otherwise). Thus, the best strategy for the online algorithm is to buy a maximum weight feasible subset of stocks with respect to the weight function $(\mu - X^t)$. This strategy is stated formally as \Cref{alg:Online}. 
Therefore, conditioned on the current prices of the stocks $X^t$, the expected profit of the algorithm resulting from buying stocks in the current time step and selling them in the next step is given by $\texttt{top}_{\mathcal{M}} \left(\mu-X^t\right)$. Since $X^t$ is identically distributed as $X$, the unconditional expected profit of the algorithm resulting from buying stocks in any time step and selling them in the next step is given by $\mathbb{E}[\texttt{top}_\mathcal{M} \left(\mu-X\right)]$. Summing over time steps $t\in[n-1]$, we get the result.
\end{proof}

\begin{algorithm}[t]
\DontPrintSemicolon
  \KwIn{Matroid $\mathcal{M} = ([k], \mathcal{I})$, joint CDF $F:\mathbb{R}^k\longrightarrow[0,1]$, and $n$.}
  Compute $\mu$, the expectation of the distribution on $\mathbb{R}^k$ whose joint CDF is $F$.\;
  \For{$t = 1$ \KwTo $n$}{
        Read the stock prices $X^t=(X^t_1,X^t_2,\dots,X^t_k)$ (drawn from the joint distribution whose CDF is $F$).\;
        Sell all currently held stocks at the prices $X^t$.\;
        $H \gets$ maximum weight feasible set of $\mathcal{M}$ with respect to the weight function $\mu - X^t$ (compute $H$ using \Cref{alg:Kruskal}).\;
        Buy all stocks in $H$.\;
  }
   \caption{Optimal online algorithm for \MatroidSetting{}.}
   \label{alg:Online}
\end{algorithm}

As a consequence of Lemma~\ref{lem_perday_opt} and Lemma~\ref{lem_perday_alg}, the task of proving bounds on competitive ratio reduces to proving bounds on the ratio of $\mathbb{E}\left[\texttt{top}_\mathcal{M} \left(\mu-X \right)\right]$ to $\mathbb{E}\left[\texttt{top}_{\mathcal{M}} \left(X' - X \right)\right] $. Here, the quantities $\mathbb{E}\left[\texttt{top}_\mathcal{M} \left(\mu-X \right)\right]$ and $\mathbb{E}\left[\texttt{top}_{\mathcal{M}} \left(X' - X \right)\right] $ are the per-time-step expected profit of the optimal online algorithm and the optimal offline algorithm, respectively. 
Similarly, for the \ProblemTwo{}, we have,

\begin{corollary}[to Lemma~\ref{lem_perday_opt}]
\label{lem_k_l_opt}
    For every instance of \ProblemTwo{}, the expected overall profit of the optimal offline algorithm is $$(n-1)\mathbb{E}\left[\texttt{\textup{top}}_{\ell'} \left(X' - X \right)\right].$$
\end{corollary}

\begin{corollary}[to Lemma~\ref{lem_perday_alg}]
\label{lem_k_l_alg}
    For every instance of \ProblemTwo{}, the expected overall profit of the optimal online algorithm is $$(n-1) \, \mathbb{E}\left[\texttt{\textup{top}}_{\ell} \left( \mu - X \right) \right].$$
\end{corollary}

\section{Hardness Results}
\label{sec:hardness_result}
\subsection{Arbitrary matroids, correlated distributions.}
In this section, we prove an upper bound on the competitive ratio of the \MatroidSetting{}.
\begin{theorem}[Restated Theorem~\ref{hardness_res_matroid}]
\label{hardness_proof}
    Let $\mathcal{M}$ be an arbitrary matroid and $d$ be the density of $\mathcal{M}$. Then there does not exist an algorithm for \MatroidSetting{} whose competitive ratio is greater than $\frac{1}{1+d}$.
\end{theorem}
\begin{proof}
    Using \Cref{lem_perday_opt}, \Cref{lem_perday_alg}, it is sufficient to construct an instance which forces  
\[\mathbb{E}\left[\texttt{top}_{\mathcal{M}} \left(\mu- X \right)\right] \leq \frac{1}{1+d} \cdot \mathbb{E}\left[\texttt{top}_{\mathcal{M}} \left(X' - X\right)\right]\text{.}\]

Since the density of $\mathcal{M}$ is $d$, there exists a subset $E'$ of the ground set $E$ of $\mathcal{M}$ such that $\frac{|E'|}{\texttt{rk}(E')} = d$. In our instance, the weights of all elements in the set $E \setminus E'$ are zero with probability one. Consequently, our analysis can be focused on the restriction of $\mathcal{M}$ to the set $E'$. Let $k = |E'|$ and $r = \texttt{rk}(E')$, which implies $d = \nicefrac{k}{r}$.

Let $\varepsilon$ be an arbitrarily small positive constant. Consider the instance in which $X$ is distributed as follows,
\[   
    X = (X_1, X_2, \ldots, X_k) =
     \begin{cases}
       (0,0,\ldots, 0) &\quad\text{w.p. } 1 - \varepsilon\text{,} \\
       (-\frac{1}{\varepsilon}, -\frac{1}{\varepsilon}, \ldots, -\frac{1}{\varepsilon}) &\quad\text{w.p. } \varepsilon - \frac{k \, \varepsilon^2 }{1 + k \varepsilon}\text{,}\\
       (\frac{1}{\varepsilon^2}, 0, \ldots, 0) &\quad\text{w.p. } \frac{ \varepsilon^2 }{1 + k \varepsilon}\text{,}\\
       \vdots &\quad \vdots\\
       (0, \ldots, 0, \frac{1}{\varepsilon^2}) &\quad\text{w.p. }  \frac{ \varepsilon^2 }{1 + k \varepsilon}\text{.}
     \end{cases}
\]
Note that $\mu=\mathbb{E}[X] = (0, \ldots, 0)$. Recall that $X'$ and $X$ are independent and identically distributed. Let $\mathcal{E'}$ be the event that $X'$ takes the value $(0, \ldots, 0)$ and $\mathcal{E}$ be the event that $X$ takes the value $(0, \ldots, 0)$. Notice that the events $\mathcal{E}$ and $\mathcal{E'}$ are independent.
The per-time-step expected profit of the optimal offline algorithm is 
 \begin{align*}
    \mathbb{E}\left[\texttt{top}_{\mathcal{M}} \left( X' - X \right)\right] &\geq \mathbb{E}\left[\texttt{top}_{\mathcal{M}} \left(X' - X \right) | \mathcal{E}\right] \cdot \Pr[ \mathcal{E}] + \mathbb{E}\left[\texttt{top}_{\mathcal{M}} \left(X' - X \right) | \mathcal{E'}\right] \cdot \Pr[ \mathcal{E'}]\\
    &= \mathbb{E}\left[\texttt{top}_{\mathcal{M}} \left(X'\right)\right] \cdot (1-\varepsilon) + \mathbb{E}\left[\texttt{top}_{\mathcal{M}} \left(- X \right)\right] \cdot (1-\varepsilon)\\
    &= \frac{1}{\varepsilon^2} \cdot \frac{k \varepsilon^2}{1+k \varepsilon} \cdot (1-\varepsilon) + \frac{r}{\varepsilon} \cdot \left(\varepsilon - \frac{k \varepsilon^2}{1 + k \varepsilon} \right) \cdot (1-\varepsilon) \\
    &= (1-\varepsilon)\frac{k}{1+k \varepsilon} + r (1-\varepsilon) \left(1 - \frac{k \varepsilon}{1 + k \varepsilon} \right)\text{.}
\end{align*}
And therefore, $ \lim_{\varepsilon \to 0} \mathbb{E}\left[\texttt{top}_{\mathcal{M}} \left( X' - X \right)\right] \geq k + r$.

Let $\mathcal{E}_1$ be the event that $X$ takes the value $(-\frac{1}{\varepsilon}, \ldots, -\frac{1}{\varepsilon})$. Except for the event $\mathcal{E}_1$, the coordinates of $X$ are all non-negative, and therefore, the weight of the maximum weight feasible set with respect to the weight function $\mu - X = -X$ is $0$.
Thus, the per-time-step expected profit of the optimal offline algorithm is
\[\mathbb{E}\left[\texttt{top}_{\mathcal{M}} \left(\mu-X\right)\right]=\mathbb{E}\left[\texttt{top}_{\mathcal{M}} \left(-X \right) | \mathcal{E}_1\right] \cdot \Pr[ \mathcal{E}_1]=\frac{r}{\varepsilon} \cdot \left(\varepsilon - \frac{k \varepsilon^2}{1 + k \varepsilon} \right)=r \left(1 - \frac{k \varepsilon}{1 + k \varepsilon} \right)\text{,}\]
and therefore, $\lim_{\varepsilon \to 0} \mathbb{E}\left[\texttt{top}_{\mathcal{M}} \left(\mu-X\right)\right] = r$.

It follows that
\[\lim_{\varepsilon \to 0} \frac{\mathbb{E}\left[\texttt{top}_{\mathcal{M}} \left(\mu-X\right)\right]}{\mathbb{E}\left[\texttt{top}_{\mathcal{M}} \left( X' - X \right)\right]} \leq \frac{r}{k+r} = \frac{1}{d+1} \text{,}\]
as required.
\end{proof}

\begin{corollary}[to Theorem~\ref{hardness_proof}]
\label{coro_hardness}
    There does not exist an algorithm for \Problem{} whose competitive ratio is greater than $\frac{\ell}{k + \ell}$.
\end{corollary}
\begin{proof}
    Follows from applying Theorem~\ref{hardness_proof} to the $\ell$-uniform matroid over the set of $k$ stocks and using Observation~\ref{observation1}.
\end{proof}

\subsection{Uniform matroids, independent distributions.}

In this section, we prove an upper bound on the competitive ratio of the \ProblemTwo{}.
\begin{theorem}[Restated Theorem~\ref{hardness_result}]
        For every $k$, ${\ell}$ and $\ell'$, there exists an independently distributed instance (refer \Cref{independently_distributed}) of the \ProblemTwo{} for which no algorithm can achieve a competitive ratio greater than $\min\left\{\frac{1}{2},\frac{\ell}{k}\right\}$.
\end{theorem}

\begin{proof}
Using \Cref{lem_k_l_opt}, \Cref{lem_k_l_alg}, it is sufficient to construct an instance which forces   
\[\mathbb{E}\left[\texttt{top}_\ell \left(\mu - X\right)\right] \leq \min\left\{\frac{1}{2},\frac{\ell}{k}\right\} \cdot \mathbb{E}\left[\texttt{top}_{\ell'} \left(X' - X \right)\right]\text{.}\]

To force the upper bound of $\nicefrac{\ell}{k}$ on the competitive ratio, consider the instance in which each $X_s$ is distributed as,
    \[   
    X_s = 
     \begin{cases}
       0 &\quad\text{w.p. } 1 - \varepsilon\text{,} \\
       \frac{1}{\varepsilon} &\quad\text{w.p. } \varepsilon\text{,}\\
     \end{cases}
\]
where $\varepsilon$ is an arbitrarily small positive constant. Recall that each $X'_s$ is identically distributed as $X_s$. Let $\mathcal{E}$ be the event that there exists a unique $s^*\in[k]$ such that $X'_{s^*}$ takes the value $\frac{1}{\varepsilon}$, and the other $2k-1$ random variables $X_s$ and $X'_s$ take value $0$, which implies $\Pr[\mathcal{E}]=k \varepsilon (1-\varepsilon)^{2k-1}$. Thus, the per-time-step expected profit of the optimal offline algorithm is bounded as follows.

\[\mathbb{E}\left[\texttt{top}_{\ell'} \left(X' - X \right)\right] \geq \mathbb{E}\left[\texttt{top}_{\ell'} \left(X' - X \right) | \mathcal{E}\right] \cdot \Pr[ \mathcal{E}]= \frac{1}{\varepsilon } \cdot k \varepsilon (1-\varepsilon)^{2k-1} = k (1-\varepsilon)^{2k-1}\text{.}\]

In this instance, for every $s \in [k]$, $\mu_s = \mathbb{E}[X_s] = 1$. Therefore, the per-time-step expected profit of the optimal online algorithm is bounded as follows.
    \[\mathbb{E}\left[\texttt{top}_\ell \left(\mu - X\right)\right] = \mathbb{E}\left[\texttt{top}_\ell \left(\mathbbm{1} - X\right)\right] \leq \mathbb{E}\left[\texttt{top}_\ell \left(\mathbbm{1}  \right)\right] = \ell\text{,}\]
    where $\mathbbm{1}$ is the all ones vector. The above inequality holds because each $X_s$ is non-negative with probability one. It follows that
    \[\frac{\mathbb{E}\left[\texttt{top}_\ell \left(\mu - X\right)\right]}{\mathbb{E}\left[\texttt{top}_{\ell'} \left(X' - X \right)\right]} \leq  \frac{\ell}{k (1-\varepsilon)^{2k-1}}\text{,}\]
    which approaches $\nicefrac{\ell}{k}$ as $\varepsilon$ approaches $0$.

To force the upper bound of $\nicefrac{1}{2}$ on the competitive ratio, consider the instance in which each $X_s$ is distributed as,
\[   
    X_s = 
     \begin{cases}
       -\frac{1}{\varepsilon} &\quad \text{w.p. } \varepsilon \text{,}\\
       0 &\quad\text{w.p. } 1 - 2\varepsilon \text{,}\\
       \frac{1}{\varepsilon} &\quad\text{w.p. } \varepsilon \text{,}\\
     \end{cases}
\]
where $\varepsilon$ is an arbitrarily small positive constant. Each $X'_s$ is distributed identically as $X_s$.
Let $\mathcal{E}$ be the event that exactly one $X_s$ takes the value $-\frac{1}{\varepsilon}$ and all other $2k-1$ random variables are $0$. Let $\mathcal{E'}$ be the event that exactly one $X_s'$ takes the value $\frac{1}{\varepsilon}$ and all other $2k-1$ random variables are $0$. Notice that the events $\mathcal{E}$ and $\mathcal{E'}$ are disjoint, and therefore, $\Pr [\mathcal{E} \cup \mathcal{E'}] = \Pr [\mathcal{E}] + \Pr [\mathcal{E}'] = k \varepsilon (1-2 \varepsilon)^{2k-1} + k \varepsilon (1-2 \varepsilon)^{2k-1} = 2 k \varepsilon (1-2 \varepsilon)^{2k-1} $. The per-time-step expected profit of the optimal offline algorithm is bounded as follows.

\[\mathbb{E}\left[\texttt{top}_{\ell'} \left(X' - X \right)\right] \geq \mathbb{E}\left[\texttt{top}_{\ell'} \left(X' - X \right) | \mathcal{E} \cup \mathcal{E'} \right] \cdot \Pr[\mathcal{E} \cup \mathcal{E'} ]= \frac{1}{\varepsilon } \cdot 2 k \varepsilon (1-2 \varepsilon)^{2k-1} = 2 k (1-2 \varepsilon)^{2k-1}\text{.}\]

Let $\mathcal{E}_1$ be the event that exactly one $X_s$ takes the value $-\frac{1}{\varepsilon}$ and all other $k-1$ coordinates of $X$ are $0$. Let $\mathcal{E}_0$ be the event that all $X_s'$ are $0$. Notice that the events $\mathcal{E}_1$ and $\mathcal{E}_0$ are disjoint. Here $\Pr [\mathcal{E}_1] = k \varepsilon (1-2 \varepsilon)^{2k-1} \leq k \varepsilon$ and $\Pr [\lnot(\mathcal{E}_0 \cup \mathcal{E}_1)] \leq \binom{k}{2} \cdot \varepsilon^2$ by union bound.
In this instance, for every $s \in [k]$, $\mu_s = \mathbb{E}[X_s] = 0$. Therefore, the per-time-step expected profit of the optimal online algorithm is bounded as follows.
\begin{align*}
\mathbb{E}\left[\texttt{top}_\ell \left(\mu - X\right)\right] &=\mathbb{E}\left[\texttt{top}_\ell \left( - X \right) | \mathcal{E}_0\right] \cdot \Pr[ \mathcal{E}_0]
+\mathbb{E}\left[\texttt{top}_\ell \left( - X \right) | \mathcal{E}_1\right] \cdot \Pr[\mathcal{E}_1]  \\
&+ \mathbb{E}\left[\texttt{top}_\ell \left( - X \right) | \lnot(\mathcal{E}_0 \cup \mathcal{E}_1)\right] \cdot \Pr[\lnot(\mathcal{E}_0 \cup \mathcal{E}_1)] \\
&\leq 0 \cdot \Pr[ \mathcal{E}_0] + \frac{1}{\varepsilon} \cdot k \varepsilon + \mathbb{E}\left[\texttt{top}_\ell \left( \frac{1}{\varepsilon}, \ldots,  \frac{1}{\varepsilon} \right) \right] \cdot \binom{k}{2} \cdot \varepsilon^2\\
&= k + \frac{l}{\varepsilon} \cdot \binom{k}{2} \cdot \varepsilon^2 =  k + l \varepsilon \cdot \binom{k}{2}\text{,}
\end{align*}
where the inequality holds because each $X_s \geq -\frac{1}{\varepsilon}$ with probability one. It follows that
\[\frac{\mathbb{E}\left[\texttt{top}_\ell \left(\mu - X\right)\right]}{\mathbb{E}\left[\texttt{top}_{\ell'} \left(X' - X \right)\right]} \leq  \frac{k + l \varepsilon \cdot \binom{k}{2}}{2 k (1-2 \varepsilon)^{2k-1}}\text{,}\]
which approaches $\nicefrac{1}{2}$ as $\varepsilon$ approaches $0$.
\end{proof}

\section{Algorithmic Results}
\label{sec:algorithmic}

For analysis, we consider a subclass of instances called \textit{zero-expectation instances} and argue that it is sufficient to prove a competitiveness guarantee of \Cref{alg:Online} on such zero-expectation instances.

\begin{definition}[zero-expectation instance]
    An instance of the \MatroidSetting{} with joint distribution having CDF $F$ is called a \textbf{zero-expectation instance} if $\mathbb{E}_{X \sim F}[X] = \bar{0}$ (the zero vector in $\mathbb{R}^k$).
\end{definition}
 
\begin{proposition}
\label{zeroExp}
    If \Cref{alg:Online} is $\alpha$-competitive on zero-expectation instances, then it is $\alpha$-competitive on all instances.
\end{proposition}

\begin{proof}
Given an arbitrary instance with joint CDF $F:\mathbb{R}^k \longrightarrow [0, 1]$,
let $\mu=(\mu_1,\ldots,\mu_k)\in\mathbb{R}^k$ be the vector of expectations of random variables whose joint CDF is $F$.
Construct an instance with joint CDF $G$ given by $G(x)=F(x+\mu)$ for all $x \in \mathbb{R}^k$. Recall that the joint CDFs of $X, X'$ are both $F$. Define $Y=X-\mu$ and $Y'=X'-\mu$, so that $X'-X=Y'-Y$. Observe that the joint CDF of $Y$ and $Y'$ is $G$, and their expectation is $\bar{0}$. Thus, the constructed instance is a zero-expectation instance.

Since $\texttt{top}_{\mathcal{M}} \left(X' - X \right)=\texttt{top}_{\mathcal{M}} \left(Y' - Y \right)$, the per-time-step profits of the offline optimal algorithms for the given and the constructed instance are the same with probability one. Since $\texttt{top}_\mathcal{M} \left(\mu- X \right)=\texttt{top}_\mathcal{M} \left(- Y \right)=\texttt{top}_\mathcal{M} \left(\mathbb{E}[Y]- Y \right)$, the per-time-step profits of \Cref{alg:Online} on the given and the constructed instance are the same with probability one. Thus, if \Cref{alg:Online} is $\alpha$-competitive on the constructed zero-expectation instance, it is $\alpha$-competitive on the given instance too.
\end{proof}

As a consequence of the above proposition, the task of getting a competitive guarantee on arbitrary instances reduces to getting a competitive guarantee on zero-expectation instances. Therefore, we assume henceforward that the instance given to \Cref{alg:Online} is a \emph{zero expectation instance}.

\subsection{Arbitrary matroids, correlated distributions.}

In this section, we prove the competitive bound for the \MatroidSetting{}. We first prove a result that quantifies and generalizes the following intuitive argument: if the matroid is not too dense then most of its elements are present in a maximum weight feasible set with respect to a non-negative weight function, and therefore, the weight of a maximum weight feasible set is close to the total weight of all the elements of the matroid.

\begin{lemma}
\label{lemma:matroid2}
Let $M$ be a loopless matroid of density $d$ on the ground set $E$, and let $w: E \longrightarrow \mathbb{R}$ be a weight function. Then the following inequality holds,
   \[\sum_{e\in E} w(e)^+ \leq d \cdot \texttt{\textup{top}}_{\mathcal{M}}(w) \text{.}\]
\end{lemma}
\begin{proof}
Let $e_1, \ldots, e_n$ be the elements of $E$ in non-increasing order of weights, that is, $w(e_1)\geq\cdots\geq w(e_n)$. Let $w(e_{n+1})=0$. Since $w(e_1)^+\geq\cdots\geq w(e_n)^+\geq w(e_{n+1})^+=0$,
we have,
\begin{equation}\label{eqn_matroid1}
\sum_{e\in E} w(e)^+ = \sum_{i=1}^n i \cdot (w(e_i)^+ - w(e_{i+1})^+) \text{.}
\end{equation}
By applying the definition of density to the set $\{e_1, \ldots, e_i\}$, we get $i \leq d \cdot \texttt{rk}(\{e_1, \ldots, e_i\})$. Thus, 
    \begin{align}
\sum_{i=1}^n i \cdot (w(e_i)^+ - w(e_{i+1})^+) &\leq d\cdot\sum_{i=1}^n \texttt{rk}(\{e_1, \ldots, e_i\}) \cdot (w(e_i)^+ - w(e_{i+1})^+)\nonumber\\
&= d\cdot\sum_{i=1}^n \left(\texttt{rk}(\{e_1, \ldots, e_i\}) - \texttt{rk}(\{e_1, \ldots, e_{i-1}\})\right) \cdot w(e_i)^+\text{,}\label{eqn_matroid2}
   \end{align}
    where the notation $\{e_1,...,e_0\}$ denotes the empty set, whose rank is zero. The ranks of the sets $\{e_1, \ldots, e_i\}$ and $\{e_1, \ldots, e_{i-1}\}$ differ by at most one, because the sets differ by just one element, namely $e_i$. If Kruskal's algorithm picks $e_i$, then $w(e_i)$ is non-negative and $e_i$ is not spanned by $\{e_1, \ldots, e_{i-1}\}$, resulting in $\texttt{rk}(\{e_1, \ldots, e_i\}) - \texttt{rk}(\{e_1, \ldots, e_{i-1}\}) = 1$. Thus, the contribution of $e_i$ to the weight of the output of Kruskal's algorithm is $w(e_i) = \left(\texttt{rk}(\{e_1, \ldots, e_i\}) - \texttt{rk}(\{e_1, \ldots, e_{i-1}\})\right) \cdot w(e_i)^+$. On the other hand, if Kruskal's algorithm does not pick $e_i$, then it means that $w(e_i)$ is non-positive or $e_i$ is spanned by $\{e_1, \ldots, e_{i-1}\}$, resulting in $\texttt{rk}(\{e_1, \ldots, e_i\}) - \texttt{rk}(\{e_1, \ldots, e_{i-1}\}) = 0$. Thus the contribution of $e_i$ to the weight of the output of Kruskal's algorithm is $0 = \left(\texttt{rk}(\{e_1, \ldots, e_i\}) - \texttt{rk}(\{e_1, \ldots, e_{i-1}\})\right) \cdot w(e_i)^+$, nevertheless. Since Kruskal's algorithm is guaranteed to find a maximum weight feasible set, we have,
    \begin{equation}\label{eq_matroid3}
        \sum_{i=1}^n \left(\texttt{rk}(\{e_1, \ldots, e_i\}) - \texttt{rk}(\{e_1, \ldots, e_{i-1}\})\right) \cdot w(e_i)^+ = \texttt{top}_{\mathcal{M}}(w^+) \text{,}
    \end{equation}
    where $w^+$ is same as $w$ except that $w^+$ assigns weight $0$ to an element if $w$ assigns a negative weight to it. Since elements with non-positive weight are ignored by Kruskal's algorithm, we have,
    \begin{equation}\label{eq_matroid4}
        \texttt{top}_{\mathcal{M}}(w^+) = \texttt{top}_{\mathcal{M}}(w) \text{.}
    \end{equation}
    Inequality (\ref{eqn_matroid2}) and \Cref{eqn_matroid1,eq_matroid3,eq_matroid4} together imply,
    \[\sum_{i=1}^n w(e_i)^+ \leq d \cdot \texttt{top}_{\mathcal{M}}(w) \text{,}\]
    as required.    
\end{proof}

Consider two arbitrary weight assignments, $x$ and $x'$, to the elements of a matroid $\mathcal{M}$. We next prove a helpful upper bound on $\texttt{top}_{\mathcal{M}}(x' - x)$.
\begin{lemma}
\label{lemma:matroid1}
    Let $\mathcal{M}$ be an arbitrary matroid over $[k]$ with density $d$. Then for every $x, x' \in \mathbb{R}^k$, the following inequality holds
    \[\texttt{\textup{top}}_{\mathcal{M}}(x' - x) \leq \texttt{\textup{top}}_{\mathcal{M}}(-x) + \sum_{i=1}^{k} x'_i + d \cdot \texttt{\textup{top}}_{\mathcal{M}}(-x') \text{.}\]
\end{lemma}
\begin{proof}
    
    Let $Z$ be a maximum weight feasible set with respect to the weight function $x' - x$. Then
\[\texttt{top}_{\mathcal{M}}(x' - x) = \sum_{i \in Z} \left(x'_i - x_i\right) = \sum_{i \in Z} x'_i + \sum_{i \in Z} \left(- x_i\right) \leq \sum_{i \in Z} x'_i + \texttt{top}_{\mathcal{M}}(-x) \text{.}\]
    Here, $\sum_{i \in Z} x'_i \leq \sum_{i \in Z} (x'_i)^+ \leq \sum_{i=1}^{k} (x'_i)^+$. Observe that every real number $y$ can be written as $y = y^+ - (-y)^+$, and therefore, $\sum_{i=1}^{k} (x'_i)^+$ can be written as $\sum_{i=1}^{k} x'_i + \sum_{i=1}^{k} (-x'_i)^+$. Thus, we get,
\[\texttt{top}_{\mathcal{M}}(x' - x) \leq \texttt{top}_{\mathcal{M}}(-x) + \sum_{i=1}^{k} x'_i + \sum_{i=1}^{k} (-x'_i)^+ \leq \texttt{top}_{\mathcal{M}}(-x) + \sum_{i=1}^{k} x'_i + d \cdot \texttt{top}_{\mathcal{M}}(-x')\text{,}
\]
    where the last inequality follows from \Cref{lemma:matroid2}.
\end{proof}

Having gathered all the necessary prerequisites, we are now prepared to present Theorem~\ref{thm:matroid:ratio}, which ensures the desired competitive ratio.

\begin{theorem}[Restated Theorem~\ref{thm:matroid:ratio}]
\label{matroid_theorem_proof}
    For every matroid $\mathcal{M}$, \Cref{alg:Online} is $\frac{1}{1+d}$-competitive, where $d$ is the density of  $\mathcal{M}$.
\end{theorem}
\begin{proof}
Recall from \Cref{lem_perday_opt} and \Cref{lem_perday_alg} that it suffices to prove
    \[\mathbb{E}\left[\texttt{top}_{\mathcal{M}} \left( X' - X \right)\right] \leq \left(d + 1\right) \mathbb{E}\left[\texttt{top}_{\mathcal{M}} \left(-X\right)\right]\text{.}\]
    By \Cref{lemma:matroid1}, the following holds with probability one:
    \[\texttt{top}_{\mathcal{M}}(X' - X) \leq \texttt{top}_{\mathcal{M}}(-X) + \sum_{i=1}^{k} X'_i + d \cdot \texttt{top}_{\mathcal{M}}(-X') \text{.}\]
    Taking expectations on both sides, we get, 
    \[ \mathbb{E} \left[\texttt{top}_{\mathcal{M}}(X' - X) \right] \leq \mathbb{E}\left[\texttt{top}_{\mathcal{M}}(-X)\right]  + \sum_i \mathbb{E}\left[ X_i' \right] + d \cdot \mathbb{E}\left[\texttt{top}_{\mathcal{M}}(-X')\right] \text{.}\]
    Observe that $\mathbb{E}[\texttt{top}_{\mathcal{M}}\left(-X\right)] = \mathbb{E}[\texttt{top}_{\mathcal{M}}\left(-X'\right)]$ because $X$ and $X'$ are identically distributed. Since we are working with zero-expectation instances, $\mathbb{E}[X_i'] = 0$ for every $i$. Therefore, we get,
    
    \[\mathbb{E}\left[\texttt{top}_{\mathcal{M}} \left( X' - X \right)\right] \leq \left(d + 1\right) \mathbb{E}\left[\texttt{top}_{\mathcal{M}} \left(-X\right)\right]\text{,}\]
as required.
\end{proof}

\begin{corollary}[to Theorem~\ref{matroid_theorem_proof}]
\label{coro_alg}
    For every $k$ and $l$, \Cref{alg:Online} is $\frac{\ell}{k + \ell}$-competitive for the \Problem{}.
\end{corollary}
\begin{proof}
    Follows from applying Theorem~\ref{matroid_theorem_proof} to the $\ell$-uniform matroid over the set of $k$ stocks and using Observation~\ref{observation1}.
\end{proof}

\subsection{Uniform matroids, independent distributions.}

In this section, we will prove that our algorithm for the \ProblemTwo{} has competitive ratio at least $\min\left\{\frac{1}{2},\frac{\ell}{k}\right\}$. We start by bounding $\mathbb{E}\left[\texttt{top}_{\ell'} \left( X' - X \right)\right]$, the per-time-step expected profit of the offline optimal algorithm in \Cref{lem:OPT} given below.

\begin{lemma}\label{lem:OPT}
For every zero-expectation instance,
\[\mathbb{E}\left[\texttt{\textup{top}}_{\ell'} \left( X' - X \right)\right] \leq \int_{0}^{\infty} \left( 2 \sum_s \Pr[-X_s \geq x] - \frac{2}{k-1} \sum_{s<s'} \Pr[-X_s \geq x] \Pr[-X_{s'} \geq x] \right) \, d x\text{.}\]
\end{lemma}

\begin{proof}
From the definition of the $\texttt{top}_{\ell'}$ function, we have,
\[\mathbb{E}\left[\texttt{top}_{\ell'} \left( X' - X \right)\right] \leq\mathbb{E}\left[\sum_{s=1}^k (X_s' - X_s)^+\right]=\mathbb{E} \left[\sum_{s=1}^k \int_{-\infty}^{\infty} \mathbbm {1}[X_s' \geq x \geq X_s]\, d x \right]\text{,}\]
where $\mathbbm{1}[\cdot]$ denotes the indicator function. By linearity of expectation and the fact that the expectation of the indicator of an event is the probability of the event, we have,
\[\mathbb{E}\left[\texttt{top}_{\ell'} \left( X' - X \right)\right]\leq\sum_{s=1}^k \left( \int_{-\infty}^{\infty} \Pr[X_s' \geq x \geq X_s] \, d x \right)=\sum_{s=1}^k \left( \int_{-\infty}^{\infty} \Pr[X_s' \geq x] \, \Pr[X_s \leq x] \, d x \right)\text{,}\]
where the equality holds because $X_s$ and $X'_s$ are independent. Since $X_s$ and $X'_s$ are identically distributed, we have,
\begin{align}
\mathbb{E}\left[\texttt{top}_{\ell'} \left( X' - X \right)\right] &\leq \sum_{s=1}^k \left( \int_{-\infty}^{\infty} \Pr[X_s \geq x] \, \Pr[X_s \leq x] \, d x \right)\nonumber\\
&= \sum_{s=1}^k \left( \int_{-\infty}^{0} \Pr[X_s \geq x] \, \Pr[X_s \leq x] \, d x + \int_{0}^{\infty} \Pr[X_s \geq x] \, \Pr[X_s \leq x] \, d x \right)\nonumber\\
&\leq \sum_{s=1}^k \left( \int_{-\infty}^{0} \left(1 - \Pr[X_s \leq x]\right) \Pr[X_s \leq x] \, d x + \int_{0}^{\infty} \Pr[X_s \geq x] \, d x \right)\text{.}\label{eqn_opt1}
\end{align}
Since we are given a \emph{zero-expectation instance}, we have for all $s\in[k]$, 
\[0=\mu_s=\int_0^{\infty}\Pr[X_s \geq x]\, d x-\int_{-\infty}^0\Pr[X_s \leq x]\, d x\text{,}\]
and thus,
\[\int_{0}^{\infty} \Pr[X_s \geq x] \, d x = \int_{-\infty}^{0}\Pr[X_s \leq x]\, d x\text{.}\]
Substituting in \Cref{eqn_opt1}, we have,
\begin{align}
\mathbb{E}\left[\texttt{top}_{\ell'} \left( X' - X \right)\right] &\leq\sum_{s=1}^k \left( \int_{-\infty}^{0} \left(1 - \Pr[X_s \leq x]\right) \Pr[X_s \leq x] \, d x + \int_{-\infty}^{0}\Pr[X_s \leq x]\, d x \right)\nonumber\\
&= \int_{-\infty}^{0} \sum_s \left( 2 \, \Pr[X_s \leq x] - \Pr[X_s \leq x]^2\right) \, d x\text{.}\label{eqn_opt2}
\end{align}
Now, we have
\[0\leq\sum_{1\leq s<s'\leq k}(\Pr[X_s \leq x]-\Pr[X_{s'} \leq x])^2=(k-1)\sum_{s\in[k]}\Pr[X_s \leq x]^2-2\sum_{1\leq s<s'\leq k}\Pr[X_s \leq x]\cdot\Pr[X_{s'} \leq x]\text{,}\]
and thus,
\[\sum_{s\in[k]}\Pr[X_s \leq x]^2\geq\frac{2}{k-1}\sum_{1\leq s<s'\leq k}\Pr[X_s \leq x]\cdot\Pr[X_{s'} \leq x]\text{.}\]
Substituting in \Cref{eqn_opt2}, we get,
\begin{align*}
    \mathbb{E}\left[\texttt{top}_{\ell'} \left( X' - X \right)\right] &\leq  \int_{-\infty}^{0} \left( 2 \sum_s \Pr[X_s \leq x] - \frac{2}{k-1} \sum_{s<s'} \Pr[X_s \leq x] \Pr[X_{s'} \leq x] \right)\,d x\\
    &= \int_{0}^{\infty} \left( 2 \sum_s \Pr[X_s \leq -x] - \frac{2}{k-1} \sum_{s<s'} \Pr[X_s \leq -x] \Pr[X_{s'} \leq -x] \right)\,d x\\
    &= \int_{0}^{\infty} \left( 2 \sum_s \Pr[-X_s \geq x] - \frac{2}{k-1} \sum_{s<s'} \Pr[-X_s \geq x] \Pr[-X_{s'} \geq x] \right) \, d x\text{,}
\end{align*}
as required.
\end{proof}

In \Cref{lem:ALG}, we derive an expression for $\mathbb{E}\left[\texttt{top}_\ell \left(\mu-X \right)\right]$, the pre-time-step expected profit of the online \Cref{alg:Online} when the stock prices at any given time are independent.

\begin{lemma}
\label{lem:ALG}
For every independently distributed zero-expectation instance, the following equality holds:
\[ \mathbb{E}\left[\texttt{\textup{top}}_\ell \left(\mu-X \right)\right] = \int_{0}^{\infty} \sum_{S \subseteq[k]} \min \{|S|, \ell\} \, \left(\prod_{s \in S} \Pr[-X_s \geq x] \right) \cdot \left(\prod_{s' \notin S} \left(1 - \Pr[-X_{s'} \geq x]\right) \right) \, d x\text{.}\]
\end{lemma}

\begin{proof}
In a \emph{zero-expectation instance}, we have $\mu_s=0$ for all $s\in[k]$. Therefore, we have,
\[\mathbb{E}\left[\texttt{top}_\ell \left(\mu-X \right)\right]=\mathbb{E}\left[\texttt{top}_\ell \left( -X \right)\right]\text{.}\]
Define random variables $M_1,\ldots,M_{\ell}$ as follows: $M_i$ is the $i$'th largest value in the multi-set $\{-X_s\mid s\in[k]\}$ if this value is positive, and $0$ otherwise. Observe that $\texttt{top}_\ell \left( -X \right)=\sum_{i=1}^{\ell}M_i$. Therefore,
\[\mathbb{E}\left[\texttt{top}_\ell \left(-X \right)\right]=\mathbb{E}\left[\sum_{i=1}^{\ell}M_i\right]=\sum_{i=1}^\ell \mathbb{E} \left[ M_i\right]= \sum_{i=1}^\ell \int_{0}^{\infty} \Pr \left[ M_i \geq x\right] \, d x=\int_{0}^{\infty} \sum_{i=1}^\ell \Pr \left[ M_i \geq x\right] \, d x\text{,}\]
where the second equality holds by linearity of expectation, and the third one holds because $M_i$ is non-negative with probability one. Using the fact that the probability of an event is the expectation of its indicator, we get,
\begin{equation}\label{eqn_alg}
\mathbb{E}\left[\texttt{top}_\ell \left(-X \right)\right]=\int_{0}^{\infty} \sum_{i=1}^\ell \mathbb{E} \left[ \mathbbm{1} \left[ M_i \geq x\right] \right]\, d x= \int_{0}^{\infty} \mathbb{E} \left[ \sum_{i=1}^\ell \mathbbm{1} \left[ M_i \geq x\right] \right]\, d x\text{,}
\end{equation}
where the second equality holds by linearity of expectation. For $x\in\mathbb{R}_+$, let the set-valued random variable $\mathcal{S}(x)$ be defined as $\mathcal{S}(x) = \{s\in[k]\mid -X_s\geq x\}$. Since $X_s$'s are independent, we have that for every $S\subseteq[k]$,
\[\Pr[\mathcal{S}(x)=S]=\left(\prod_{s \in S} \Pr[-X_s \geq x]\right)\cdot\left(\prod_{s' \notin S} \left(1 - \Pr[-X_{s'} \geq x]\right)\right)\text{.}\]
Conditioned on the event $\mathcal{S}(x)=S$, the event $M_i\geq x$ happens if and only if $|S|\geq i$, and hence, $\sum_{i=1}^\ell \mathbbm{1} \left[ M_i \geq x\right]=\min\{|S|,\ell\}$ with probability one. Using this fact and the expression for $\Pr[\mathcal{S}(x)=S]$ in \Cref{eqn_alg}, we get,
\begin{align*}
\mathbb{E}\left[\texttt{top}_\ell \left(-X \right)\right] &= \int_{0}^{\infty} \sum_{S\subseteq[k]}\mathbb{E} \left[ \sum_{i=1}^\ell \mathbbm{1} \left[ M_i \geq x\right] \bigg| \mathcal{S}(x)=S\right]\cdot\Pr[\mathcal{S}(x)=S]\, d x\\
&= \int_{0}^{\infty} \sum_{S \subseteq[k]} \min \{|S|, \ell\} \, \left(\prod_{s \in S} \Pr[-X_s \geq x] \right) \cdot \left(\prod_{s' \notin S} \left(1 - \Pr[-X_{s'} \geq x]\right) \right) \, d x\text{,}
\end{align*}
as required.
\end{proof}

Next, we derive an inequality that will be used to relate $\mathbb{E}\left[\texttt{top}_{\ell'} \left( X' - X \right)\right]$ and $\mathbb{E}\left[\texttt{top}_\ell \left(\mu - X \right)\right]$ in Theorem~\ref{theorem1}.

\begin{lemma}
\label{inequality1}
    For every $k$, $\ell$ and $\{a_i\}_{i \in [k]} \in [0,1]^k$, the following inequality holds:
    \[ \max \left\{2, \frac{k}{\ell} \right\} \cdot \left(\sum_{S \subseteq [k]} \min \{|S|, \ell\} \, \prod_{s \in S} a_s \prod_{s' \notin S} \left(1 - a_{s'}\right) \right) \geq \left(2 \, \sum_{s\in[k]} a_s - \frac{2}{k-1}\sum_{s,s'\in[k]\text{, }s<s'} a_s a_{s'} \right)\text{.} \]
\end{lemma}
\begin{proof}
    Define $f:[0,1]^k\longrightarrow\mathbb{R}$ as
\begin{multline*}
f(a_1,\ldots, a_k)=\max \left\{2, \frac{k}{\ell} \right\} \cdot \left(\sum_{S \subseteq [k]} \min \{|S|, \ell\} \, \prod_{s \in S} a_s \prod_{s' \notin S} \left(1 - a_{s'}\right) \right) \\ - \left(2 \, \sum_{s\in[k]} a_s - \frac{2}{k-1}\sum_{s,s'\in[k]\text{, }s<s'} a_s a_{s'} \right)\text{.}
\end{multline*}
We need to prove that $f$ is non-negative on $[0,1]^n$. Since $f$ is multilinear in $a_1,\ldots,a_k$, it is sufficient to prove that $f$ is non-negative on the corners of its domain $[0,1]^n$, that is, $f(a_1,\ldots,a_k)\geq0$ for all $(a_1,\ldots, a_k)\in\{0,1\}^k$ (see Corollary 2.3 of \cite{laneve2010interval}).

Consider an arbitrary $(a_1,\ldots, a_k)\in\{0,1\}^k$. Let $z$ be the number of $a_s$'s are 1 so that $k-z$ is the number of the $a_s$'s are 0. Thus $z\in\{0,\ldots,k\}$. Also,
\[f(a_1,\ldots, a_k)=\max \left\{2, \frac{k}{\ell} \right\}\cdot\min \{z, \ell\} - 2z + \frac{2}{k-1}\cdot\binom{z}{2} = \max \left\{2, \frac{k}{\ell} \right\}\cdot\min \{z, \ell\} - 2z + \frac{z(z-1)}{k-1}\text{.} \] 
If $z\leq\ell$, then
\[f(a_1,\ldots,a_k)=\max \left\{2, \frac{k}{\ell} \right\}\cdot z - 2z + \frac{z(z-1)}{k-1}\geq\frac{z(z-1)}{k-1}>0\text{.}\]
If $z>\ell$, then 
\[f(a_1,\ldots,a_k)=\max \left\{2, \frac{k}{\ell} \right\}\cdot \ell - 2z + \frac{z(z-1)}{k-1}\geq\frac{k}{\ell}\cdot \ell - 2z + \frac{z(z-1)}{k-1}=k-\frac{z(2k-z-1)}{k-1}\text{.}\]
It is easy to check that the function $z\longmapsto k-z(2k-z-1)/(k-1)$ is $0$ at $z=k$ and $z=k-1$, and is decreasing in $[0,k-1]$. Thus $k-z(2k-z-1)/(k-1)\geq0$ for all $z\in[k]$, which implies $f(a_1,\ldots,a_k)\geq0$, as required.
\end{proof}
Having gathered all the necessary prerequisites, we are now prepared to present Theorem~\ref{theorem1}, which ensures the desired competitive ratio.
\begin{theorem}[Restated Theorem~\ref{theorem1}]
    For every $k$, $\ell$, and $\ell'$, \Cref{alg:Online} is a $\min\left\{\frac{1}{2},\frac{\ell}{k}\right\}$-competitive algorithm for the \ProblemTwo{} when the instance is independently distributed (refer \Cref{independently_distributed}).
\end{theorem}

\begin{proof}
    In \Cref{lem:OPT}, we obtained the following upper bound on the per-time-step expected profit of offline algorithm:
     \[\mathbb{E}\left[\texttt{top}_{\ell'} \left( X' - X \right)\right] \leq \int_{0}^{\infty} \left( 2 \sum_s \Pr[-X_s \geq x] - \frac{2}{k-1} \sum_{s<s'} \Pr[-X_s \geq x] \Pr[-X_{s'} \geq x] \right) \, d x \text{.}\]
    In \Cref{lem:ALG}, we obtained the following expression for the per-time-step expected profit of \Cref{alg:Online}:
     \[ \mathbb{E}\left[\texttt{top}_\ell \left(\mu - X \right)\right] = \int_{0}^{\infty} \sum_{S \subseteq[k]} \min \{|S|, \ell\} \, \left(\prod_{s \in S} \Pr[-X_s \geq x] \right) \cdot \left(\prod_{s' \notin S} \left(1 - \Pr[-X_{s'} \geq x]\right) \right) \, d x \text{.}\] 
For an arbitrary $x\geq0$, using \Cref{inequality1} with $a_s = \Pr [-X_s \geq x]$ for all $s \in [k]$, we get,
\begin{multline*}
   \max \left\{2, \frac{k}{\ell} \right\} \cdot \sum_{S \subseteq[k]} \min \{|S|, \ell\} \, \left(\prod_{s \in S} \Pr[-X_s \geq x] \right) \cdot \left(\prod_{s' \notin S} \left(1 - \Pr[-X_{s'} \geq x]\right) \right) \\
 \geq 2 \sum_s \Pr[-X_s \geq x] - \frac{2}{k-1} \sum_{s<s'} \Pr[-X_s \geq x] \Pr[-X_{s'} \geq x]\text{.}
\end{multline*}
Integrating both sides of the above inequality with respect to $x$ from $0$ to $\infty$, the left-hand-side is $\max \left\{2, \frac{k}{\ell} \right\}$ times $\mathbb{E}\left[\texttt{top}_\ell \left(\mu - X \right)\right]$, and the right-hand-side is the upper bound on $\mathbb{E}\left[\texttt{top}_{\ell'} \left( X' - X \right)\right]$. Therefore, we get,
\[\max \left\{2, \frac{k}{\ell} \right\}\cdot\mathbb{E}\left[\texttt{top}_\ell \left(\mu - X \right)\right] \geq \mathbb{E}\left[\texttt{top}_{\ell'} \left( X' - X \right)\right]\text{.}\]
Thus, \Cref{alg:Online} is $\min\left\{\frac{1}{2},\frac{\ell}{k}\right\}$-competitive.
\end{proof}

\section{Non-i.i.d.\ Random order Trading Prophet Problem}
\label{sec:random_order}
An instance of the non-i.i.d.\ random order matroid trading prophet is defined by a matroid $\mathcal{M}$ on the ground set $\{1, \ldots, k\}$ for some $k \in \mathbb{N}$, and $n$ (possibly distinct) joint distributions of stock prices over $\mathbb{R}^k$. Price vectors drawn independently from these $n$ distributions are presented in a uniformly random order to a trading algorithm.  
As before, the algorithm holds a feasible set of $\mathcal{M}$ at all times, and at each time step, it must decide which subset of currently held stocks to sell and which subset of stocks not currently held to buy. Again, not selling a held stock is equivalent to selling it and repurchasing it at the same price within the same time step. Therefore, we may assume that at every time step, the algorithm sells all its currently held stocks, and then selects a subset $S$ to buy, ensuring that $S$ is a feasible set of $\mathcal{M}$.

Let $F^1, \ldots, F^n$ denote the joint CDFs of $n$ distributions over $\mathbbm{R}^k$. Let $\sigma = (\sigma(1), \ldots, \sigma(n))$ be a uniformly random permutation of the set $\{1, \ldots, n\}$ and let $X^1, \ldots, X^n$ be the sequence of the revealed prices, where each $X^t$ is the price vector drawn independently from the distribution whose CDF is $F^{\sigma(t)}$.

Consider the equivalent setup where at time step $t$, the algorithm sells all currently held stocks and then buys a subset of stocks. Suppose it buys a subset $S$ of stocks at time $t$ and sells it at time $t+1$. Its profit will be $\sum_{s \in S} (X^{t+1}_s-X^t_s)$, and therefore, at every time $t$, the optimal algorithm buys a maximum weight feasible subset of stocks with respect to the weight function $X^{t+1}-X^t$. Thus, the expression for the total expected profit of the optimal offline algorithm (denoted as OPT) is
\[\text{OPT} = \sum_{t=1}^{n-1} \mathbb{E} \left[\texttt{top}_{\mathcal{M}} \left(X^{t+1} - X^t\right) \right]\text{.}\]
Observe that for each $t$, $\sigma(t)$ is uniformly distributed over $\{1,\ldots,n\}$, while conditioned on $\sigma(t)$, $\sigma(t+1)$ is uniformly distributed over $\{1,\ldots,n\}\setminus\{\sigma(t)\}$. Thus, for each $t$, the pair $(X^{t+1}, X^t)$ is identically distributed as the pair $(X^2, X^1)$. Consequently, we have,
\begin{align*}
    \text{OPT} = (n-1) \, \mathbb{E} \left[\texttt{top}_{\mathcal{M}} \left(X^{2} - X^1\right) \right]\text{.}
\end{align*}
Therefore, the per-time-step expected profit of the offline optimal algorithm is $\mathbb{E} \left[\texttt{top}_{\mathcal{M}} \left(X^{2} - X^1\right) \right]$.

To motivate further analysis, we first describe our online algorithm informally. Observe that the random variables $X^1,\ldots,X^n$ are identically distributed, and the marginal distribution of each of them is the mixture of the $n$ distributions specified by CDFs $F^1,\ldots,F^n$. However, these random variables are not independent. We plan to reuse \Cref{alg:Online} with its parameter $F$ being the CDF of each of the $X^t$'s (i.e., $F=(F^1+\cdots+F^n)/n$). However, the competitive guarantee of \Cref{alg:Online} holds only when the online price vectors passed to it are sampled independently from a fixed distribution. Therefore, in order to make use of the competitive guarantee, it becomes necessary to relate the actual profit of the online algorithm as well as the optimal offline algorithm with the corresponding profit if the price vectors were to be sampled independently from the distribution with CDF $F$.

Let us first relate the actual per-time-step expected profit of the optimal offline algorithm with the per-time-step expected profit if the prices were to be independent. In other words, we relate $\mathbb{E} \left[\texttt{top}_{\mathcal{M}} \left(X^2 - X^1\right) \right]$ with $\mathbb{E} \left[\texttt{top}_{\mathcal{M}} \left(\Tilde{X}^2 - X^1\right) \right]$, where $\Tilde{X}^2$ is an appropriately constructed random variable independent of and identically distributed as $X^1$. To facilitate analysis, we construct the random variable $\Tilde{X}^2$ in such a way that it is heavily correlated with $X^2$, as follows. We toss a biased coin that lands heads with probability $1/n$ and tails with probability $1-1/n$. If the coin lands heads, we set $\Tilde{X}^2$ to be an independent sample drawn from the distribution whose joint CDF is $F^{\sigma(1)}$. If the coin lands tails, we set $\Tilde{X^2} = X^2$. Thus, we have,

\begin{Observation}
\label{obs_unconditional}
$X^1$ and $X^2$ are samples from two \textbf{distinct} distributions out of $F^1, \ldots, F^n$ chosen uniformly at random. In contrast, $X^1$ and $\Tilde{X^2}$ are samples from two \textbf{not necessarily distinct} distributions out of $F^1, \ldots, F^n$ chosen uniformly at random, and are, therefore, independent and identically distributed.
\end{Observation}

\begin{lemma}
\label{lem_matroid_random_1}
\[\mathbb{E}\left[\texttt{\textup{top}}_{\mathcal{M}}\left(\Tilde{X^2} - X^1 \right)\right] \geq \frac{n-1}{n}\cdot\mathbb{E}\left[\texttt{\textup{top}}_{\mathcal{M}}\left( X^2 - X^1 \right)\right]\text{.}\]
\end{lemma}

\begin{proof}
    Let $H$ denote the event that the outcome of the coin tossed to generate $\Tilde{X}^2$ is heads, and let $T$ denote the event that the outcome is tails. Then we have,
    \begin{align*}
        \mathbb{E}\left[\texttt{top}_{\mathcal{M}}\left( \Tilde{X^2} - X^1 \right)\right] &= \mathbb{E}\left[\texttt{top}_{\mathcal{M}}\left( \Tilde{X^2} - X^1 \right) \middle|  H \right] \cdot \Pr[H] + \mathbb{E}\left[\texttt{top}_{\mathcal{M}}\left( \Tilde{X^2} - X^1 \right) \middle| T \right] \cdot \Pr[T] \\
        &\geq \mathbb{E}\left[\texttt{top}_{\mathcal{M}}\left( \Tilde{X^2} - X^1 \right) \middle| T \right] \cdot \Pr[T] \\
        &= \left(1 - \frac{1}{n}\right) \, \mathbb{E}\left[\texttt{top}_{\mathcal{M}}\left( X^2 - X^1 \right) \right]\text{,}
    \end{align*}
as required, where the inequality holds because $\texttt{top}_{\mathcal{M}}\left( \Tilde{X^2} - X^1 \right)$ is a non-negative random variable.
\end{proof}

As stated earlier, our online algorithm for non-i.i.d.\ random order trading prophet is \Cref{alg:Online} run with the parameter $F = (F^1 + \cdots + F^n)/n$, the marginal CDF of each $X^t$. Let $\mu$ be the expectation of each $X^t$. Our online algorithm is stated formally as \Cref{alg:rand}. For analysis, let $\mu^i$ denote the expectation of the distribution with CDF $F^i$, and let $\overline{\mu}^i$ denote the expectation of the mixture of $n-1$ distributions whose CDFs are $\{F^1, \ldots, F^n\}\setminus\{F_i\}$. Thus, we have
\begin{equation}\label{eqn_mu1}
\mu=\frac{\sum_{j\in[n]}\mu^j}{n}\text{,}
\end{equation}
and $\overline{\mu}^i=(\sum_{j\in[n]\setminus\{i\}}\mu^j)/(n-1)$, which imply that
\begin{equation}\label{eqn_mu2}
\mu =\frac{\mu^i+(n-1)\overline{\mu}^i}{n}\text{.}
\end{equation}

\begin{algorithm}[t]
\DontPrintSemicolon
  \KwIn{Matroid $\mathcal{M} = ([k], \mathcal{I})$, joint CDFs $F^1,\ldots,F^n:\mathbb{R}^k\longrightarrow[0,1]$.}
  Compute $\mu$, the expectation of the distribution on $\mathbb{R}^k$ whose joint CDF is $(F^1+\cdots+F^n)/n$.\;
  \For{$t = 1$ \KwTo $n$}{
        Read the stock price vector $X^t=(X^t_1,X^t_2,\dots,X^t_k)$.\;
        Sell all currently held stocks at the price vector $X^t$.\;
        $H^t \gets$ maximum weight feasible set of $\mathcal{M}$ with respect to the weight function $\mu - X^t$.\;
        Buy all stocks in $H^t$ at the price vector $X^t$.\;
  }
   \caption{Online algorithm for non-i.i.d.\ random order trading prophet problem.}
   \label{alg:rand}
\end{algorithm}

Consider an arbitrary time $t$ at which the online algorithm receives the price vector $X^t$ sampled from $F^{\sigma(t)}$. After selling all the held stocks, the algorithm buys a subset $H^t$ of stocks that is feasible in $\mathcal{M}$ and that maximizes the weight function $\mu-X^t$. Thus, $\sum_{j\in H^t}(\mu_j-X^t_j)=\texttt{top}_{\mathcal{M}}\left( \mu - X^t \right)$, and this is exactly the profit the algorithm expects to obtain at the next instant of time, if it were to run on i.i.d.\ price vectors drawn from the mixture distribution. However, in reality, conditioned on the value $\sigma(t)$, the distribution of $X^{t+1}$ is the mixture of all $F^i$'s except $F^{\sigma(t)}$, and its expectation is $\overline{\mu}^{\sigma(t)}$, not $\mu$. Thus, conditioned on $\sigma(t)$ and $X^t$, the algorithm's expected profit from selling the set $H^t$ of stocks at time $t+1$ is $\sum_{j\in H^t}(\overline{\mu}^{\sigma(t)}_j-X^t_j)$. Therefore, the expected total profit of the online algorithm is given by

\[\text{ALG} = \sum_{t=1}^{n-1}\mathbb{E}_{\sigma(t),X^t}\left[ \sum_{j \in H^t}(\overline{\mu}^{\sigma(t)}_j - X_j^t) \right]\text{.}\]
Observe that the random variables $\sum_{j \in H^t}(\overline{\mu}^{\sigma(t)}_j - X_j^t)$ are identically distributed, and therefore,
\[\text{ALG} = (n-1) \cdot \mathbb{E}_{\sigma(1),X^1}\left[ \sum_{j \in H^1}(\overline{\mu}^{\sigma(1)}_j - X_j^1) \right]\text{.}\]
Thus, the per-time-step expected profit of \Cref{alg:rand} is $\mathbb{E}_{\sigma(1),X^1}\left[ \sum_{j \in H^1}(\overline{\mu}^{\sigma(1)}_j - X^1_j) \right]$. Our goal is to relate this quantity to $\mathbb{E}_{X^1}\left[\texttt{\textup{top}}_{\mathcal{M}} \left( \mu - X^1 \right)\right]=\mathbb{E}_{X^1}\left[ \sum_{j \in H^1}(\mu_j - X^1_j) \right]$, the per-time-step expected profit of \Cref{alg:Online} if the price vectors were to be drawn independently from the mixture distribution. Towards this, we first show that for a large $n$, the random variable $\overline{\mu}^i$ for a uniformly random $i$ drawn from $[n]$ is close to its expectation $\mu$ in $1$-norm. Specifically,

\begin{lemma}
    \label{lem_matroid_random_3}
    \[
    \mathbb{E}_{i\sim[n]} \left[\sum_{j\in[k]} |\mu_j - \overline{\mu}^i_j| \right] \leq \frac{2d}{n-1}\cdot\mathbb{E}_{X^1}\left[ \texttt{\textup{top}}_{\mathcal{M}}  \left( \mu - X^1 \right) \right].
    \]
\end{lemma}
\begin{proof}
Rearranging equation (\ref{eqn_mu2}), we get,
    \[\mu - \overline{\mu}^i= \frac{\mu^i - \mu}{n-1}\text{.}\]
Therefore,
\begin{equation}\label{eqn_1}
\mathbb{E}_{i\sim[n]} \left[\sum_{j\in[k]} |\mu_j - \overline{\mu}^i_j| \right]= \frac{ \mathbb{E}_{i\sim[n]}\left[ \sum_{j\in[k]} \left|\mu^i_j - \mu_j \right| \right]}{n-1}= \frac{  \sum_{j\in[k]}\mathbb{E}_{i\sim[n]}\left[ \left|\mu^i_j - \mu_j \right| \right]}{n-1}\text{.}
\end{equation}
Observe that $\left|\mu^i_j - \mu_j \right|=\left(\mu_j - \mu^i_j \right)^++\left(\mu^i_j - \mu_j \right)^+$. From equation (\ref{eqn_mu1}), we have $\mu=\mathbb{E}_{i\sim[n]}[\mu^i]$, which implies 
\[\mathbb{E}_{i\sim[n]}\left[ \left(\mu_j - \mu^i_j \right)^+-\left(\mu^i_j - \mu_j \right)^+ \right]=\mathbb{E}_{i\sim[n]}\left[\mu_j - \mu^i_j\right]=0\text{.}\]
Thus,
\begin{align*}
\mathbb{E}_{i\sim[n]}\left[ \left|\mu^i_j - \mu_j \right| \right] &= \mathbb{E}_{i\sim[n]}\left[ \left(\mu_j - \mu^i_j \right)^++\left(\mu^i_j - \mu_j \right)^+ \right]\\
 &= \mathbb{E}_{i\sim[n]}\left[ \left(\mu_j - \mu^i_j \right)^++\left(\mu^i_j - \mu_j \right)^+ \right]+\mathbb{E}_{i\sim[n]}\left[ \left(\mu_j - \mu^i_j \right)^+-\left(\mu^i_j - \mu_j \right)^+ \right]\\
 &= 2\cdot\mathbb{E}_{i\sim[n]}\left[ \left(\mu_j - \mu^i_j \right)^+\right]\text{.}
\end{align*}
Substituting in equation (\ref{eqn_1}), we get,
\begin{equation}\label{eqn_2}
\mathbb{E}_{i\sim[n]} \left[\sum_{j\in[k]} |\mu_j - \overline{\mu}^i_j| \right]=\frac{2\cdot  \sum_{j\in[k]}\mathbb{E}_{i\sim[n]}\left[ \left(\mu_j - \mu^i_j \right)^+ \right]}{n-1}\text{.}
\end{equation}
Applying \Cref{lemma:matroid2} with the weight function $w = \mu - \mu^i$, we get,
\begin{equation}\label{eqn_3}
\mathbb{E}_{i\sim[n]}\left[ \sum_{j\in[k]} \left(\mu_j - \mu^i_j \right)^+ \right]\leq d\cdot\mathbb{E}_{i\sim[n]}\left[ \texttt{top}_{\mathcal{M}}  \left( \mu - \mu^i \right)  \right]\text{,}
\end{equation}
where $d$ is the density of the matroid $\mathcal{M}$. Let $Y^i$ be a maximum weight feasible set of $\mathcal{M}$ with respect to the weight function $\mu-\mu^i$. Then we have,
    \begin{align}
        \mathbb{E}_{i\sim[n]}\left[ \texttt{top}_{\mathcal{M}}  \left( \mu - \mu^i \right)  \right] &= \mathbb{E}_{i\sim[n]}\left[ \sum_{y \in Y^i} \left( \mu_y - \mu^i_y \right) \right]\nonumber\\
        &= \frac{1}{n}\cdot\sum_{i=1}^n\sum_{y \in Y^i} \left( \mu_y - \mu^i_y \right)\nonumber\\
        &= \frac{1}{n}\cdot\sum_{i=1}^n\sum_{y \in Y^i}  \mathbb{E}_{X^1} \left[ (\mu_y - X^1_y) \middle| \sigma(1) = i\right]\nonumber\\
        &= \frac{1}{n}\cdot\sum_{i=1}^n\mathbb{E}_{X^1} \left[ \sum_{y \in Y^i} (\mu_y - X^1_y) \middle| \sigma(1) = i\right]\nonumber\\
        &\leq \frac{1}{n}\cdot\sum_{i=1}^n\mathbb{E}_{X^1} \left[\texttt{top}_{\mathcal{M}} \left(\mu - X^1\right)\middle| \sigma(1) = i\right]\nonumber\\
        &= \mathbb{E}_{X^1} \left[\texttt{top}_{\mathcal{M}} \left(\mu - X^1\right) \right]\text{.}\label{eqn_4}
    \end{align}
Equation (\ref{eqn_2}) and inequalities (\ref{eqn_3}) and (\ref{eqn_4}) together imply,
\[
\mathbb{E}_{i\sim[n]} \left[\sum_{j\in[k]} |\mu_j - \overline{\mu}^i_j| \right] \leq \frac{2d}{n-1}\cdot\mathbb{E}_{X^1}\left[ \texttt{\textup{top}}_{\mathcal{M}}  \left( \mu - X^1 \right) \right]\text{,}
\]
as required.
\end{proof}

We use the above lemma to bound the discrepancy between the per-time-step expected profit of \Cref{alg:rand} and the per-time-step expected profit of the optimal online algorithm (\Cref{alg:Online}) if the prices were to be drawn independently from the mixture distribution as follows.

\begin{lemma}
\label{lem_matroid_random_2}
    \[
    \mathbb{E}_{\sigma(1),X^1}\left[\sum_{j \in H^1} \left( \overline{\mu}^{\sigma(1)}_j - X^1_j \right)\right] \geq \left(1 - \frac{2 d}{n-1}\right) \mathbb{E}_{X^1}\left[\texttt{\textup{top}}_{\mathcal{M}} \left( \mu - X^1 \right)\right].
    \]
\end{lemma}
\begin{proof}
    For each stock $j \in [k]$, we have
    \[
    \mu_j - X_j^1 \leq (\mu_j - X_j^1) + (\overline{\mu}^{\sigma(1)}_j - \mu_j) + |\overline{\mu}^{\sigma(1)}_j - \mu_j| = (\overline{\mu}^{\sigma(1)}_j - X_j^1) + |\overline{\mu}^{\sigma(1)}_j - \mu_j|.
    \]
    Summing over all $ j \in H^1 $, we obtain
    \[
    \sum_{j \in H^1}(\mu_j - X_j^1) \leq \sum_{j \in H^1}(\overline{\mu}^{\sigma(1)}_j - X_j^1) + \sum_{j \in H^1}|\overline{\mu}^{\sigma(1)}_j - \mu_j|\text{.}
    \]
    Recall from the definition of \Cref{alg:rand} that $H^1$ is a maximum weight feasible set of the matroid $\mathcal{M}$ with respect to the weight function $\mu-X^1$. Therefore, $\sum_{j \in H^1}(\mu_j - X_j^1) = \texttt{top}_{\mathcal{M}} \left( \mu - X^1 \right)$. Moreover, $\sum_{j \in H^1}|\overline{\mu}^{\sigma(1)}_j - \mu_j|\leq\sum_{j\in[k]}|\overline{\mu}^{\sigma(1)}_j - \mu_j|$. Thus,
\[\texttt{top}_{\mathcal{M}} \left( \mu - X^1 \right)\leq \sum_{j \in H^1}(\overline{\mu}^{\sigma(1)}_j - X_j^1) + \sum_{j\in[k]} |\overline{\mu}^{\sigma(1)}_j - \mu_j|\text{.}\]
Taking the expectation on both sides over $\sigma(1),X^1$, we get,
    \begin{align*}
        \mathbb{E}_{X^1}\left[\texttt{top}_{\mathcal{M}} \left( \mu - X^1 \right) \right] &\leq \mathbb{E}_{\sigma(1),X^1} \left[\sum_{j \in Z}(\overline{\mu}^{\sigma(1)}_j - X_j^1) \right] + \mathbb{E}_{\sigma(1)}\left[\sum_{j\in[k]} |\overline{\mu}^{\sigma(1)}_j - \mu_j|\right]\\
        &\leq \mathbb{E}_{\sigma(1),X^1} \left[\sum_{j \in Z}(\overline{\mu}^{\sigma(1)}_j - X_j^1) \right] + \frac{2d}{n-1}\cdot\mathbb{E}_{X^1}\left[ \texttt{\textup{top}}_{\mathcal{M}}  \left( \mu - X^1 \right) \right],
    \end{align*}
    where the last inequality follows from \Cref{lem_matroid_random_3}, because $\sigma(1)$ is uniformly distributed over $[n]$. Rearranging the terms, we get:
    \[
    \mathbb{E}_{\sigma(1),X^1} \left[\sum_{j \in Z}(\overline{\mu}^{\sigma(1)}_j - X_j^1) \right] \geq \left(1 - \frac{2 d}{n-1}\right) \mathbb{E}_{X^1}\left[\texttt{\textup{top}}_{\mathcal{M}} \left( \mu - X^1 \right)\right]\text{,}
    \]
    as required.   
\end{proof}

Having gathered all the necessary prerequisites, we now present the proof of the competitive ratio of our algorithm.

\begin{theorem}[Restated Theorem~\ref{thm:random_order}]
For the non-i.i.d.\ random-order trading prophet problem over a matroid of density $d$ and with $n$ distributions, \Cref{alg:rand} achieves a competitive ratio of at least $\frac{1}{1+d} - \frac{2}{n}$.
\end{theorem}

\begin{proof}
    From \Cref{lem_matroid_random_1}, the per-time-step expected profit of the optimal offline algorithm is bounded as,
\begin{equation}\label{eqn_rand1}
    \mathbb{E}\left[\texttt{\textup{top}}_{\mathcal{M}}\left(\Tilde{X^2} - X^1 \right)\right] \geq \frac{n-1}{n}\cdot\mathbb{E}\left[\texttt{\textup{top}}_{\mathcal{M}}\left( X^2 - X^1 \right)\right]\textit{,}
\end{equation}
    where $X_1,\Tilde{X}_2$ are independent and identically distributed, each having CDF $F$. From \Cref{lem_matroid_random_2}, the per-time-step expected profit of the online algorithm is bounded as,
\begin{equation}\label{eqn_rand2}
    \mathbb{E}_{\sigma(1),X^1}\left[\sum_{j \in Z} \left( \overline{\mu}^{\sigma(1)}_j - X^1_j \right)\right] \geq \left(1 - \frac{2 d}{n-1}\right) \mathbb{E}_{X^1}\left[\texttt{\textup{top}}_{\mathcal{M}} \left( \mu - X^1 \right)\right]\text{,}
\end{equation}
    where $\mu$ is the expectation of the mixture distribution.

    Recall that our online algorithm (given by \Cref{alg:rand}) is nothing but the optimal online algorithm (given by \Cref{alg:Online}) for the i.i.d.\ trading prophet run on the mixture distribution. Consider a random experiment involving the run of \Cref{alg:Online} on the i.i.d.\ trading prophet instance, where the distribution of each price vector is the mixture distribution. In this experiment, the per-time-step expected profit of the optimal offline algorithm is $\mathbb{E}\left[\texttt{\textup{top}}_{\mathcal{M}}\left(\Tilde{X^2} - X^1 \right)\right]$, while that of the online algorithm is $\mathbb{E}_{X^1}\left[\texttt{\textup{top}}_{\mathcal{M}} \left( \mu - X^1 \right)\right]$. The competitiveness guarantee of \Cref{alg:Online} (given by \Cref{matroid_theorem_proof}) ensures,
\begin{equation}\label{eqn_rand3}
    \mathbb{E}_{X^1}\left[\texttt{\textup{top}}_{\mathcal{M}} \left( \mu - X^1 \right)\right]\geq\frac{1}{1+d}\cdot\mathbb{E}\left[\texttt{\textup{top}}_{\mathcal{M}}\left(\Tilde{X^2} - X^1 \right)\right]
\end{equation}
    Putting together inequalities (\ref{eqn_rand1}), (\ref{eqn_rand2}), (\ref{eqn_rand3}), we get,
    \[\mathbb{E}_{\sigma(1),X^1}\left[\sum_{j \in Z} \left( \overline{\mu}^{\sigma(1)}_j - X^1_j \right)\right] \geq \left(1 - \frac{2 d}{n-1}\right)\cdot\frac{1}{1+d}\cdot\frac{n-1}{n}\cdot\mathbb{E}\left[\texttt{\textup{top}}_{\mathcal{M}}\left( X^2 - X^1 \right)\right]\text{.}\]
    Thus, \Cref{alg:rand} achieves competitive ratio at least
    \[\left(1 - \frac{2 d}{n-1}\right)\cdot\frac{1}{1+d}\cdot\frac{n-1}{n}=\left(1 - \frac{2d+1}{n}\right)\cdot\frac{1}{1+d}\geq\frac{1}{1+d} - \frac{2}{n}\text{,}\]
    as required.
\end{proof}

\section{Concluding Remarks}

In this work, we formulated generalizations of the trading prophet problem, expanding the task of trading a single stock to trading multiple stocks under matroid constraints. We identified the behavior of the optimal offline and online trading strategies and, as a consequence, established the optimal competitive ratio in each case.

In the setup where the price vectors are independent and identically distributed, our generalizations were consequences of our simplification of the results of \cite{CCD+23trading} using the two observations. First, any algorithm can be simulated by one that sells all held stocks before buying an appropriate subset of stocks. Second, the analysis of the general problem reduces to one in which the expected price of every stock is zero.

When price vectors are not identically distributed and arrive in a uniformly random order, our results are consequences of the following intuitive observation. Consider two random experiments conducted on a set of size $n$. In the first experiment, two objects are sampled without replacement, while in the second, two objects are sampled with replacement -- uniformly and independently in both experiments. As $n$ approaches $\infty$, the distributions of the outcomes in both experiments become increasingly similar, because the probability of a collision in the second experiment decreases to zero. This idea enables us to reuse the algorithm for the i.i.d.\ setting with a little error analysis.

Our analysis is crucially dependent on the greedy Kruskal's algorithm being able find a maximum weight feasible set in a downward-closed set system. It is known that Kruskal's algorithm has this guarantee if and only if the underlying set system is a matroid. Thus, the design and analysis of trading algorithms over non-matroid constraints will require significantly novel ideas, and we leave this as an open problem. The matching constraint over a bipartite graph is arguably one of the simplest such constraints.

\section*{Acknowledgments}
We are grateful to anonymous reviewers for their helpful comments. RV acknowledges support from DST INSPIRE grant no. DST/INSPIRE/04/2020/000107, SERB grant no. CRG/2022/002621, and iHub Anubhuti IIITD Foundation.



\begin{thebibliography}{16}
\providecommand{\natexlab}[1]{#1}
\providecommand{\url}[1]{\texttt{#1}}
\expandafter\ifx\csname urlstyle\endcsname\relax
  \providecommand{\doi}[1]{doi: #1}\else
  \providecommand{\doi}{doi: \begingroup \urlstyle{rm}\Url}\fi

\bibitem[Alaei(2014)]{A14bayesian}
S.~Alaei.
\newblock {Bayesian Combinatorial Auctions: Expanding Single Buyer Mechanisms to Many Buyers}.
\newblock \emph{SIAM Journal on Computing}, 43\penalty0 (2):\penalty0 930--972, 2014.

\bibitem[Correa et~al.(2019)Correa, Foncea, Hoeksma, Oosterwijk, and Vredeveld]{CFJ+19recent}
J.~Correa, P.~Foncea, R.~Hoeksma, T.~Oosterwijk, and T.~Vredeveld.
\newblock {Recent Developments in Prophet Inequalities}.
\newblock \emph{ACM SIGecom Exchanges}, 17\penalty0 (1):\penalty0 61--70, 2019.

\bibitem[Correa et~al.(2023)Correa, Cristi, D{\"u}tting, Hajiaghayi, Olkowski, and Schewior]{CCD+23trading}
J.~Correa, A.~Cristi, P.~D{\"u}tting, M.~T. Hajiaghayi, J.~Olkowski, and K.~Schewior.
\newblock {Trading Prophets}.
\newblock In \emph{{EC}}, pages 490--510, 2023.

\bibitem[Ehsani et~al.(2018)Ehsani, Hajiaghayi, Kesselheim, and Singla]{EHK+18prophet}
S.~Ehsani, M.~T. Hajiaghayi, T.~Kesselheim, and S.~Singla.
\newblock {Prophet Secretary for Combinatorial Auctions and Matroids}.
\newblock In \emph{SODA}, pages 700--714. SIAM, 2018.

\bibitem[Esfandiari et~al.(2017)Esfandiari, Hajiaghayi, Liaghat, and Monemizadeh]{EJL+17prophet}
H.~Esfandiari, M.~Hajiaghayi, V.~Liaghat, and M.~Monemizadeh.
\newblock {Prophet Secretary}.
\newblock \emph{SIAM Journal on Discrete Mathematics}, 31\penalty0 (3):\penalty0 1685--1701, 2017.

\bibitem[Hajiaghayi et~al.(2007)Hajiaghayi, Kleinberg, and Sandholm]{HKS07automated}
M.~T. Hajiaghayi, R.~Kleinberg, and T.~Sandholm.
\newblock {Automated Online Mechanism Design and Prophet Inequalities}.
\newblock In \emph{AAAI}, volume~7, pages 58--65, 2007.

\bibitem[Hill and Kertz(1982)]{HillK}
T.~P. Hill and R.~P. Kertz.
\newblock {Comparisons of Stop Rule and Supremum Expectations of I.I.D. Random Variables}.
\newblock \emph{The Annals of Probability}, 10\penalty0 (2):\penalty0 336 -- 345, 1982.
\newblock \doi{10.1214/aop/1176993861}.
\newblock URL \url{https://doi.org/10.1214/aop/1176993861}.

\bibitem[Kleinberg and Weinberg(2019)]{KleinbergW19}
R.~Kleinberg and S.~M. Weinberg.
\newblock {Matroid Prophet Inequalities and Applications to Multi-dimensional Mechanism Design}.
\newblock \emph{Games Econ. Behav.}, 113:\penalty0 97--115, 2019.
\newblock \doi{10.1016/J.GEB.2014.11.002}.
\newblock URL \url{https://doi.org/10.1016/j.geb.2014.11.002}.

\bibitem[Krengel and Sucheston(1978)]{KrengelS}
U.~Krengel and L.~Sucheston.
\newblock {On Semiamarts, Amarts, and Processes with Finite Value}.
\newblock \emph{Adv. in Probability}, 4:\penalty0 197--266, 1978.

\bibitem[Laneve et~al.(2010)Laneve, Lascu, and Sordoni]{laneve2010interval}
C.~Laneve, T.~A. Lascu, and V.~Sordoni.
\newblock {The Interval Analysis of Multilinear Expressions}.
\newblock \emph{Electronic Notes in Theoretical Computer Science}, 267\penalty0 (2):\penalty0 43--53, 2010.

\bibitem[Lucier(2017)]{L17economic}
B.~Lucier.
\newblock {An Economic View of Prophet Inequalities}.
\newblock \emph{ACM SIGecom Exchanges}, 16\penalty0 (1):\penalty0 24--47, 2017.

\bibitem[Oxley(1992)]{oxley}
J.~G. Oxley.
\newblock \emph{{Matroid Theory}}.
\newblock Oxford University Press, 1992.
\newblock ISBN 978-0-19-853563-8.

\bibitem[Peng and Tang(2022)]{PT22order}
B.~Peng and Z.~G. Tang.
\newblock {Order Selection Prophet Inequality: From Threshold Optimization to Arrival Time Design}.
\newblock In \emph{FOCS}, pages 171--178. IEEE, 2022.

\bibitem[Roughgarden(2020)]{Roughgarden}
T.~Roughgarden.
\newblock \emph{{Resource Augmentation}}.
\newblock Cambridge University Press, 2020.
\newblock \doi{10.1017/9781108637435.006}.
\newblock URL \url{https://doi.org/10.1017/9781108637435.006}.

\bibitem[Sleator and Tarjan(1985)]{SleatorT85}
D.~D. Sleator and R.~E. Tarjan.
\newblock {Amortized Efficiency of List Update and Paging Rules}.
\newblock \emph{Commun. {ACM}}, 28\penalty0 (2):\penalty0 202--208, 1985.
\newblock \doi{10.1145/2786.2793}.
\newblock URL \url{https://doi.org/10.1145/2786.2793}.

\bibitem[Soto(2013)]{soto2013matroid}
J.~A. Soto.
\newblock {Matroid Secretary Problem in the Random-Assignment Model}.
\newblock \emph{SIAM Journal on Computing}, 42\penalty0 (1):\penalty0 178--211, 2013.

\end{thebibliography}
\end{document}